\documentclass[12pt]{article}

\usepackage[left=2.5cm,right=2.5cm,top=3cm,bottom=3cm,a4paper]{geometry}
\usepackage{authblk}
\usepackage[numbers]{natbib}
\usepackage{alphalph,etoolbox}
\usepackage{graphicx}
\usepackage{mathrsfs,amsmath,amssymb,amscd,mathtools,amsthm}
\usepackage{adjustbox}
\usepackage{setspace}
\usepackage{caption}
\usepackage{hyperref}
\usepackage{subfig}
\usepackage{float}
\usepackage{appendix}
\usepackage{proof}
\usepackage{enumerate}
\usepackage{xspace}
\usepackage{indentfirst}
\usepackage{tabularray}

\counterwithin*{equation}{section}
\counterwithin*{equation}{subsection}

\newcommand{\kn}{$k$-out-of-$n$\xspace}
\newcommand{\kng}{$k$-out-of-$n$: G\xspace}
\newcommand{\kngb}{$k$-out-of-$n$: G balanced\xspace}
\newcommand{\knpgb}{$k$-out-of-$n$ pairs: G balanced\xspace}
\newcommand{\ckngb}{circular $k$-out-of-$n$: G balanced\xspace}
\newcommand{\etal}{\textit{et~al.}\xspace}

\newtheorem{definition}{Definition}
\newtheorem{proposition}{Proposition}[definition]
\newtheorem{lemma}{Lemma}
\newtheorem{remark}{Remark}[definition]
\newtheorem{corollary}{Corollary}

\begin{document}
    \nocite{*}
    \title{Reliability Improvement of Circular $k$-out-of-$n$: G Balanced Systems through Center of Gravity}

    \author[1]{Yongkyu Cho}
    \author[2]{Seung Min Baik}
    \author[2]{Young Myoung Ko\thanks{Correspondence to Young Myoung Ko (youngko@postech.ac.kr)}}
    
    \affil[1]{Department of Industrial Engineering, Kangnam University, South Korea\\ E-mail: yongkyu.cho@kangnam.ac.kr}
    \affil[2]{Department of Industrial and Management Engineering, Pohang University of Science and Technology, South Korea\\E-mails: gshs27@postech.ac.kr, youngko@postech.ac.kr}

    \maketitle

    \begin{abstract}   
        This paper considers a \ckngb system equipped with homogeneous and stationary units. Building on previous research by Endharta \etal\cite{EYK2018}, we propose a new balance definition in \ckngb systems based on the concept of center of gravity. According to this condition, a \ckngb system is considered balanced if its center of gravity is located at the origin. This new balance condition is not only simple but also advantageous as it covers the previous two balance conditions of symmetry and proportionality. To evaluate the system's reliability, we consider the minimum tie-sets, and extensive numerical studies verify the enhancement of system reliability resulting from the proposed balance definition.
    \end{abstract}

    \newpage
        
    \section*{Notation}
    \begin{itemize}
        \item $n$: number of units in a system
        \item $k$: minimum number of non-failed units for a functioning system
        \item $r$: probability that a unit is functioning properly
        \item $\mathcal{U}$: index set of all units in the system; $\mathcal{U}=\{1,2,...,i,...,n\}$
        \item $U$ or $U_j$: ($j^{\textrm{th}}$) subset of $\mathcal{U}$ denoting the set of non-failed units; $U_j=\{u_1,u_2,...,u_l,...,u_{|U_j|}\}$ where each $u_l$ corresponds to a unit index $i\in \mathcal{U}$ such that $u_1<u_2<\cdots<u_{|U_j|}$.
        \item $\mathbf{U}_k$: set of $k$-combinations of set $\mathcal{U}$; $\mathbf{U}_k=\{U_1,U_2,...,U_{|\mathbf{U}_k|}\}$ where $|\mathbf{U}_k|=\binom{n}{k}.$
        \item $d_{l}$: distance between units $u_{l+1}$ and $u_{l}$ in a subset $U$ for $l=1,2,...,|U|$; $d_{l}=u_{l+1}-u_{l}$ if ${l}\neq |U|$ and $n+u_1-u_{|U|}$ if $l=|U|$.
        \item $D_U$: distance tuple for a subset $U$; $D_U=(d_1,d_2,...,d_{|U|})$
        \item $E^{(l)}$: reverse tuple of $D_U$ starting from the $l^\textrm{th}$ distance; $E^{(l)}=(d_l,d_{l-1},...,d_2,d_1,d_{|U|},$\\$d_{|U|-1},...,d_{l+1})$
        \item $\mathbf{E}_{U}^B$: set of $E^{(l)}$'s such that $E^{(l)}=D_U$ for a subset $U$; the system is balanced if $|\mathbf{E}_U^B|$ is an even number or if $|\mathbf{E}_U^B|>1$.
        \item $T$ or $T_t$: ($t^{\textrm{th}}$) tie-set of a system; $T_t=\{u_1,u_2,...,u_l,...,u_{|T_t|}\}$
        \item $\mathbf{T}$: set of all tie-sets of a system; $\mathbf{T}=\{T_1,T_2,...,T_t,...,T_{|\mathbf{T}|}\}$
        \item $T^M$ or $T_m^M$: ($m^{\textrm{th}}$) minimum tie-set of a system; $T_m^M=\{u_1,u_2,...,u_l,...,u_{|T_m^M|}\}$
        \item $\mathbf{T}^M$: set of minimum tie-sets; $\mathbf{T}^M=\{T_1^M,T_2^M,...,T_m^M,...,T_{|\mathbf{T}^M|}\}$
        \item $\mathbf{T}^N$: set of non-minimum tie-sets; $\mathbf{T}^M=\mathbf{T}\backslash\mathbf{T}^N.$
        \item $X_i$: binary state variable of unit $i$; $X_i=1$ if unit $i$ is functioning, $X_i=0$ otherwise.
        \item $\mathbf{X}$: system state vector; $\mathbf{X}\equiv[X_1,...,X_n]$
        \item $\phi(\cdot)$: system structure function; $\phi(\mathbf{X})=1$ if the system is functioning, $\phi(\mathbf{X})=0$ otherwise.
        \item $R_S$: system reliability; $R_S\equiv\mathbb{P}\left[\phi\left(\mathbf{X}\right)=1\right]=\mathbb{E}\left[\phi\left(\mathbf{X}\right)\right].$
    \end{itemize}
	
	
    \section{Introduction} \label{sec:introduction}
    This paper considers systems with spatially distributed units such as balanced engine systems in planetary descent vehicles and Unmanned Aerial Vehicles (UAVs), commonly called \textit{drones} equipped with multiple rotary-wings as shown in Fig.~\ref{fig:system_example}. In the balanced engine system, if an engine in a pair fails during landing, system dynamics mandate that the other engine in the pair must be turned off in order to uphold the balance of the descent vehicle \cite{HBS2000}. The \kn balanced system can be used to analyze the reliability in such a situation because there may be a minimum number of pairs of engines required to prevent rapid descent during landing. Similarly to the descent system, there may be a minimum number of operating rotary-wings to operate drones reliably. In fact, many systems with spatially distributed units have a circular configuration with evenly distributed units. The definition of balance varies depending on the context of the analysis, so reliability evaluation does not have a clear-cut solution. Additionally, considering the units' positions is the most significant factor that makes the problem difficult when defining the system's balance.
    
    Inspired by drones and descent systems, Hua and Elsayed \cite{HE2016a,HE2016b,HE2018}, along with a series of papers, defined this situation as a \knpgb system and conducted extensive reliability studies for this system. Unlike descent systems where each engine provides propulsion in a vertical direction to the ground, however, the rotors of a drone generate lift force while rotating, making the definition of balance more complicated. In the series of research studies regarding drone reliability, the balance condition was mainly defined based on system symmetry.
    
	
    Meanwhile, Endharta \etal \cite{EYK2018} proposed a new balance condition based on the proportionality of the system for a similar situation. That is, as long as the operating units are evenly distributed throughout the system, balance can be maintained. Furthermore, considering such balance conditions, the units of the system no longer need to necessarily form pairs, which can lead to a relaxation of the target system from a \kn pairs system to a \kn system. As a result, Endharta \etal \cite{EYK2018} investigated a reliability analysis problem for \ckngb systems. Note that a \kn pairs system can be seen as a special case of a \kn system. For example, if Fig.~\ref{fig:octacopter_model} is considered a $k$-out-of-$4$ pairs system, it can be regarded as a $2k$-out-of-$8$ system with some additional conditions. Therefore, a balance condition based on symmetry can also be applied to examine the balance of \kn systems, which can lead to a direct comparison between two different balance conditions.
	
    Building on the previous research study by Endharta \etal\cite{EYK2018}, this paper presents another new balance definition for \ckngb systems. This new balance definition is advantageous due to its simplicity and generality, as it covers the two balance conditions previously investigated. In essence, the new balance definition directly takes into account the concept of center of gravity. The main contributions of our work are summarized as follows.
    \begin{itemize}
		\item We propose a new balance definition for \ckngb systems that can enhance the system reliability.
		\item We investigate the inclusion relationships among the three different balance conditions discovered so far through mathematical proofs and numerical examples.
		\item We demonstrate the reliability improvement resulting from the new balance definition using extensive numerical analysis. We also discuss the effect of system parameters on the overall system reliability.
    \end{itemize}
	
    The remainder of this paper is organized as follows. Section \ref{sec:literature_review} conducts a literature review of existing reliability studies on a variety of \kngb systems. Section \ref{sec:system} describes the target system and outline the key modeling assumptions. Section \ref{sec:balance} presents three definitions of balance conditions, including the newly proposed one. This section also provides an in-depth investigation of the relationships among the three balance conditions. Section \ref{sec:reliability} explains the reliability evaluation method with a descriptive numerical example. Section \ref{sec:numerical} provides extensive numerical results showing the enhancement of reliability resulting from the proposed balance definition. Finally, Section \ref{sec:conclusion} concludes and suggests possible extensions of the research study.
	
    \section{Literature Review} \label{sec:literature_review}
    The \kn models can be largely divided into two types: \kn: F and \kn: G. A \kn: F system fails when at least $k$ units are failed whereas a \kn: G system functions when at least $k$ units are non-failed. In a \kng system with homogeneous units of binary-states (failed and non-failed), the number of non-failed units follows the binomial distribution; the system reliability coincides with the probability that $k$ or more units are non-failed \cite{ELSAYED2021}. 
	
    The \ckngb system, which is the target system of this paper, is a variant of \kng model in which the units are arranged circularly, and the reliable system requires to satisfy a certain balance condition. Among the variants of \ckngb models, the most relevant special case is the $k$-out-of-$n$ \textit{pairs}: G balanced system which has $n$ pairs of units distributed evenly on a circular configuration. For example, the graphical model in Fig.~\ref{fig:octacopter_model} can be considered to have four pairs of units: (1,5), (2,6), (3,7), and (4,8). Sarper and Sauer \cite{SS2002} presented the first reliability study for such a system: balanced engine systems in planetary descent vehicles, e.g., four-engine (i.e., two pairs) and six-engine (i.e., three pairs) configurations. The authors provided a basic reliability evaluation framework based on simple stochastic models such as Bernoulli distributed unit state and exponentially distributed failure time. Applying the \knpgb model to drone systems, Hua and Elsayed \cite{HE2016a,HE2016b} conducted the extensive reliability studies regarding degradation analysis and reliability estimation. Notably, Hua and Elsayed \cite{HE2016b} were the first to define the balanced state of a drone system by symmetry, considering the concept of Moment Difference to measure the degree of symmetry. The authors considered two different scenarios (rebalancing is allowed or not) and develop a systematic approach for estimating the reliability of \knpgb systems. Although the proposed approach was effective for reliability evaluation, the computational intractability when systems are large remained a limitation. In this regard, Hua and Elsayed \cite{HE2018} proposed a Monte Carlo simulation-based reliability approximation method which is fast as well as accurate hence applicable even for large systems. Considering more complex configurations of drone systems, Guo and Elsayed \cite{GE2019} presented a $(k_1,k_2)$-out-of-$(n,m)$ pairs: G balanced system. The system modeled a rotary UAV with multi-level of rotors where $m$ rotors are arranged vertically in parallel in the same position. Reliability estimations for two scenarios (forced-down rotors are considered as failed and forced down rotors are considered as standbys) have been obtained by enumerating operational states and calculating the probability of their occurrences.
	
    Unlike the above literature that only considered system symmetry as a measure of balance in drone systems and conducted extensive reliability studies, Endharta \etal \cite{EYK2018} proposed a new balance definition based on system proportionality that can be applied to drone systems. By considering the new balance definition, the target system can be relaxed from a \kn pairs to a \kn system. Endharta \etal \cite{EYK2018} considered a system with homogeneous and stationary units, and the minimum tie-sets were enumerated to evaluate the system reliability. The results showed the reliability improvement compared to when considering only the symmetry-based balance definition. In the subsequent study by Endharta and Ko \cite{EK2020}, a load-sharing system was considered where the amount of load is equally distributed among the working units. The failure time of each unit was assumed to follow an exponential distribution, and the load-sharing relationship between the amount of load and the failure rate is assumed to follow a power rule. System failure paths were arranged to evaluate reliability, and for larger systems, reliability was approximated using Monte Carlo simulation.
	
    Regarding the technique of reliability evaluation, the Inclusion-Exclusion (IE) method is a well-known approach that can be used for the systems consisting of units with constant failure probability. Using the IE method, reliability expression of such systems can be derived based on the concepts of minimum tie-sets or cut-sets \cite{ELSAYED2021}. To improve the computational efficiency of the IE method, Heidtmann \cite{H1982}  proposed an improved version that eliminates certain terms in the expression. Another popular method called the sum-of-disjoint-products (SDP) method, which is also related to minimum tie-sets or cut-sets, has been explored by McGrady \cite{M1985}. For \kng systems with independent but non-identical units, Rushdi \cite{R1986} developed a pivotal decomposition approach for evaluating reliability. Recently, Hao \etal \cite{HYWWS2019} proposed a new IE method called Quick Inclusion-Exclusion (QIE) to increase the efficiency of the IE method and reduce the amount of required memory.
	
    There have been a body of recent reliability studies on a variant of balanced systems. Wang \etal \cite{WZZ2022} studied a multi-state $k$-out-of-$n$: F balanced system with a rebalancing mechanism, which is important in engineering fields like new energy storage and aeronautics. The components and system are assumed to have multiple states, and an age maintenance strategy and optimization model were investigated to obtain the optimal results. The proposed model was demonstrated through numerical examples based on a product line balancing problem. Dui \etal \cite{DZBC2021} proposed a model for studying the mission reliability and structure optimization of an UAV swarm based on importance measures. The mission reliability model was based on polygonal linear consecutive-$k$-out-of-$n$: F systems, and the structure optimization has been analyzed using conditional reliability, conditional failure rate, and remaining useful life. The proposed method was demonstrated through numerical examples of triangular and quadrilateral UAV swarms. Wang \etal \cite{WQWL2021} presented a condition-based preventive maintenance policy for balanced systems with identical components. The system state was evaluated based on the deterioration levels of all components, and PM activities are employed to avoid competing failures. The optimal preventive maintenance thresholds are determined by minimizing the system maintenance cost within the framework of semi-Markov decision process. Zhao and Wang \cite{ZW2022} suggested a new maintenance policy optimization method for systems with two balanced components, assuming that the components degrade over time according to a bivariate Wiener process. The objective of the maintenance actions was to eliminate the differences of degradation levels of system components at the cost of aggravating the degradation, and the model was optimized using Markov decision process with both finite and infinite planning horizons. Wang \etal \cite{WZL2022} proposed a maintenance policy optimization model for balanced systems composed of multiple functionally-exchangeable units, where a unit may fail due to self-failure or external stress. The objective was to find the optimal number of operating units to minimize the maintenance cost per unit time, and an illustrative example was used to demonstrate the effectiveness of the proposed policy. Wu \etal \cite{WZWS2022} studied a load-sharing consecutive-$k$-out-of-$r$-from-$n$ subsystems: F balanced system with linear and circular structures that maintains balance by forcing working components to standby or resuming standby components to operate. The system fails if there are $r$ consecutive subsystems in which at least $k$ subsystems fail, and the system reliability was analyzed using the Markov processes and finite Markov chain imbedding approach. Tian \etal \cite{TYLW2023} presented a new reliability evaluation method for Performance-based balanced systems with common bus performance sharing (PBSs-CBPS), which considers the balance degree threshold, transmission loss, and transmission capacity limit. A continuous-time discrete-state Markov model was built to address the transition behaviors of components, and the universal generation function method combined with nonlinear programming was proposed to calculate system reliability. Wang \etal \cite{WNZW2023} investigated two reliability models for balanced systems with multi-state protective devices, considering rebalancing mechanisms and failure criteria of both dynamic and static balanced concepts. The models involved new triggering and variable protection mechanisms to reduce the impact of shocks and internal degradation on units, and the reliability was derived using the Markov process imbedding method with Monte-Carlo simulation. 
	
    In the next section, we describe the target system and introduce the key assumptions considered in this paper.
	
    \section{System Description} \label{sec:system}
    Consider a system with circularly arranged units such as descent systems equipped by $n$ engines or a drone product equipped by $n$ rotary-wings. Although the system starts up in the perfect state with $n$ operational units, some of them can go wrong as time goes by. The system with some failed units (say a \textit{subsystem}), however, still has possibility to be reliable if the number of non-failed units are enough to generate sufficient forces to remain the system in the air, for example, at least $k$ units are non-failed then the drone can stay afloat in the air. In such situation, what we must consider together is the balance of the overall system. The set of non-failed units that can keep the balance will generate the total thrust in the same direction to the lifting forces so that the system stays in the air. The \ckngb system is a reliability model that abstracts out such a situation. Fig.~\ref{fig:system_example} shows a pair of such system (a drone product) and its graphical model. 
	
    \begin{figure}[ht]
	\centering
	\subfloat[][A multi-rotor drone product with eight wings \cite{drone_product}]
	{
		\centering\resizebox{0.5\textwidth}{!}{\includegraphics{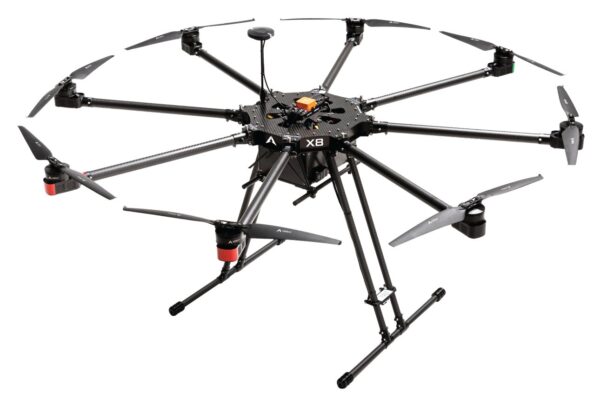}}
		\label{fig:octacopter}		
	}
	~
	\subfloat[][A circular $k$-out-of-$8$: G balanced system]
	{
		\centering\resizebox{0.4\textwidth}{!}{\includegraphics{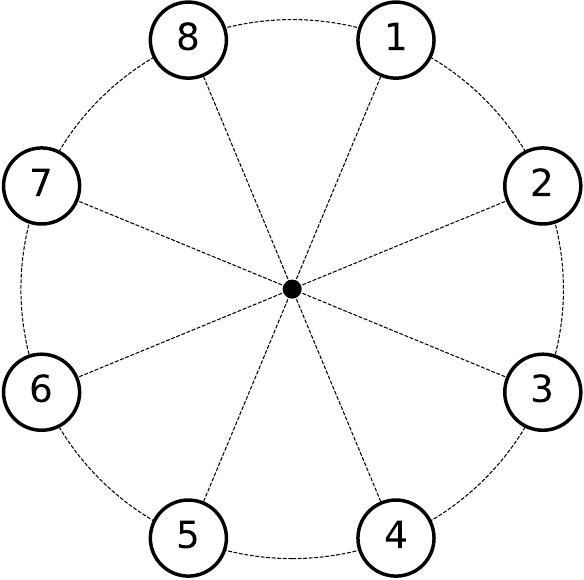}}
		\label{fig:octacopter_model}		
	}
	
	\caption{An illustrative example of a multi-rotor drone and its corresponding system model}
	\label{fig:system_example}
    \end{figure}
	
    Here, we explain some of the main notations to mathematically describe the target system. First of all, we consider a set of indexes of units $\mathcal{U}$ (${U}=\{1,2,...,n\}$) that comprise a \ckngb system. Then, $U$ (or $U_j$) denotes the ($j^\textrm{th}$) subset of $\mathcal{U}$ containing the set of indexes non-failed units (i.e., $U=\{u_1,u_2,...,u_{|U|}\}$ where each $u_l$ corresponds to each non-failed unit index). For example, if we have $\mathcal{U}=\{1,2,3,4\}$ and unit $2$ has been failed, then we have a subset $U=\{1,3,4\}$ where $u_1=1,u_2=3,u_3=4$. Note that $u_l$'s are in an ascending order: $u_1<u_2<\cdots<u_{|U|}$. We also introduce a collection of sets denoted by $\mathbf{U}_k$ that consists of all the $k$-combinations of a unit set $\mathcal{U}$, that is, $\mathbf{U}_k=\{U_1,U_2,...,U_{|\mathbf{U}_k|}\}$ where $|\mathbf{U}_k|=\binom{n}{k}$.
	
    A \ckngb system is said to be operational if at least $k$ among $n$ units are non-failed as well as the system maintains a balance with the non-failed units. Since On-Off control is not a highly sophisticated control mechanism in modern electronic devices, it is reasonable to assume each non-failed unit can be forced up and down. This On-Off mechanism can be utilized to rebalance the system by forcing down some operational units that are harming the system balance. To evaluate the system-wise reliability of such systems, it is necessary to take the unit-wise reliability into account. In this regard, we assume that the system consists of independent and homogeneous units each of which has a stationary survival probability. In short, we assume a constant survival and failure probabilities $r$ and $1-r$ for each unit and for all the planning horizon. 
	
    Below, we summarize the key assumptions considered for the last of this paper. 
    \begin{itemize}
	\item The system consists of independent and homogeneous units each of which has a constant survival and failure probabilities $r$ and $1-r$.
	\item Each non-failed unit is subject to On-Off control. For example, any non-failed operating unit can be forced down for rebalancing the system.
    \end{itemize}
	
    In the next section, we formally define the balance conditions that can be considered regarding the reliability analysis of \ckngb systems.
	
    \section{Balance Conditions} \label{sec:balance}
    We consider three balance conditions for \ckngb systems: Balance Condition I (symmetry, BC1 \cite{HE2016b}), Balance Condition II (proportionality, BC2 \cite{EYK2018}), and the newly proposed Balance Condition III (center of gravity at origin, BC3). Note that just satisfying one of the balance conditions does not necessarily mean that the system is functioning; for a \kng system to be operational, there should be at least $k$ non-failed units as well. Taken together, a \ckngb system is operational when it satisfies at least one of the balance conditions with at least $k$ non-failed units. Although this section introduces all the three balance conditions for explanatory purpose, we want to emphasize that the proposition of BC3 is the contributing part of this paper. Hence, we refer readers to the previous literature \cite{HE2016b,EYK2018,EK2020} for in-depth explanation on BC1 and BC2.
	
    \subsection{BC1: \textit{System is symmetric}} \label{subsec:bc1}
    The first balance condition BC1 is related to the symmetry of the target system. That is, if the system is symmetric in terms of a certain definition, it is considered as balanced. BC1 has been originally suggested to evaluate the circular $k$-out-of-$n$ \textit{pairs}: G balanced systems. Since the system can be regarded as a special case of \ckngb system in which an additional \textit{pair} constraint is considered, BC1 can also be applied to any of the \ckngb systems for evaluating its balance. The following definition states BC1 that is suggested by Hua and Elsayed \cite{HE2016b}.
    \begin{definition}[BC1] \label{def:bc1}
	A \ckngb system is said to be balanced if it is symmetric in the sense that all its operating units are symmetric w.r.t at least a pair of perpendicular axes as well as the number of pairs is an even number.
    \end{definition}
	
    Fig.~\ref{fig:illustrative_example_BC1} depicts illustrative examples that explain BC1. Regarding all the graphical models that will be shown in this paper, white circles express the operating units whereas black-colored circles correspond to the failed or turned-off units. First, the graphical model in Fig.~\ref{fig:ex_BC1_2_axes} describes a subsystem with $U=\{4,5,6,10,11,12\}$. As shown in the figure, it forms a pair of perpendicular axes of symmetry so that is identified to satisfy BC1 hence balanced. In comparison, Fig.~\ref{fig:ex_BC1_3_axes} provides a subsystem with $U=\{4,8,12\}$ that forms three distinct axes of symmetry in which no pair of perpendicular axes is observed. According to definition \ref{def:bc1}, it is identified not to satisfy BC1 hence unbalanced. Later in Section~\ref{subsec:bc2}, we will show that the subsystem in Fig.~\ref{fig:ex_BC1_3_axes} also can be identified to be balanced by considering more generalized balance condition (see Fig.~\ref{fig:ex_BC2_case1}).
	
    \begin{figure}[ht]
	\centering
	\subfloat[][One pair of perpendicular axes]
	{
		\centering\resizebox{0.35\textwidth}{!}{\includegraphics{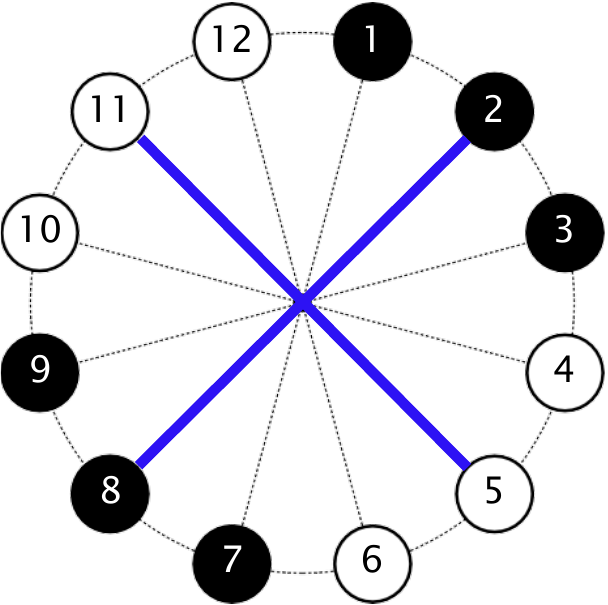}}
		\label{fig:ex_BC1_2_axes}		
	}
	~
	\subfloat[][No pair of perpendicular axes]
	{
		\centering\resizebox{0.35\textwidth}{!}{\includegraphics{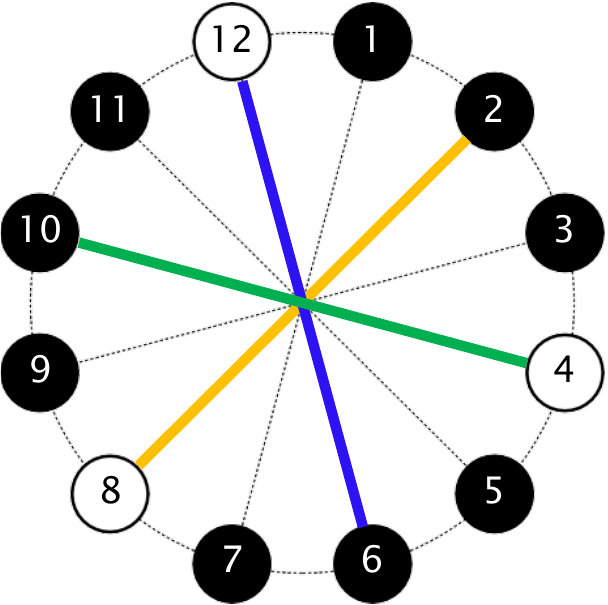}}
		\label{fig:ex_BC1_3_axes}		
	}
	\caption{Illustrative examples for BC1}
	\label{fig:illustrative_example_BC1}
    \end{figure}
	
    To explain the quantitative way of judging whether or not the system satisfies BC1, we introduce some mathematical notations. First, let $d_l$ denote the distance between units $u_{l+1}$ and $u_l$ defined by $d_l=u_{l+1}-u_l$ if $l\neq|U|$ and $d_l=n+u_1-u_{|U|}$ if $l=|U|$, for a subsystem $U=\{u_1,u_2,...,u_{|U|}\}$ where $l=1,2,...,|U|$. We also define $D_U$ be a tuple that collects all the distance enumeration: $D_U=(d_1,d_2,...,d_{|U|})$. Then, let $E^{(l)}$ be the reverse tuple of $D_U$ which starts from $d_l$; $E^{(l)}=(d_l,d_{l-1},...,d_2,d_1,d_{|U|},d_{|U|-1},...,d_{l+1})$. Note that an $E^{(l)}$ that satisfies $E^{(l)}=D_U$ corresponds to an axis of symmetry. By collecting all such $E^{(l)}$'s, we generate a set of reverse tuples denoted by $\mathbf{E}_U^B$. Combined with definition \ref{def:bc1}, the following proposition states a quantitative way of figuring out if a system satisfies BC1.
	
    \begin{proposition}[Endharta \etal \cite{EYK2018}] \label{prop:eyk2018}
    For a \ckngb system with an index set of non-failed units $U$, the axis of symmetry is located between units $u_1$ and $u_{l+1}$ for $l$ such that $E^{(l)}=D_U$. Thus, $|\mathbf{E}_U^B|$ can represent the number of axes of symmetry in the system and the system balance is examined as follows.
    \begin{enumerate}[\indent(a)]
		\item If $|\mathbf{E}_U^B|$ is an even number, the system is symmetric w.r.t. at least a pair of perpendicular axes and balanced.
		\item If $|\mathbf{E}_U^B|$ is an odd number, the system is symmetric w.r.t. certain axes, but not balanced.
		\item If $|\mathbf{E}_U^B|$ is 0, the system is not symmetric, then it is not balanced.
	\end{enumerate}
    \end{proposition} 
	
    Simply put, the result of proposition~\ref{prop:eyk2018} leads towards the following remark.
    \begin{remark} \label{remark:how_to_examing_bc1}
        A \ckngb system with an index set of non-failed units $U$ satisfies BC1 if $|\mathbf{E}_U^B|$ is a nonzero even number.
    \end{remark}
	
    In the next subsection, we will explain an expanded balance definition (BC2), which is the main result of the previous study by Endharta \etal \cite{EYK2018}.
	
    \subsection{BC2: \textit{System is spread proportionally}} \label{subsec:bc2}
    The second balance condition BC2 focuses on the proportionality of the operating units in a system. According to BC2, a \ckngb system is considered as balanced if the system is spread proportionally by the operating units. 
    \begin{definition}[BC2] \label{def:bc2}
	A \ckngb system is said to be balanced if the operating units are spread proportionally within the system.
    \end{definition}
	
    To quantify the definition~\ref{def:bc2}, we investigate the angles between two non-consecutively neighboring operating units (say \textit{sector angle}), and then examine the patterns observed from the sector angles. 	Fig.~\ref{fig:illustrative_example_BC2} shows some representative examples of the systems satisfying BC2 and the following remark states the quantitative way of evaluating BC2 proposed by Endharta \etal \cite{EK2020}.
	
    \begin{remark}
    Let $a_s$ be the $s^\textrm{th}$ sector angle between the two non-consecutively neighboring operating units and $N$ be the total number of such angles. Then, a \ckngb system showing at least one of the following patterns is said to be balanced.
	\begin{enumerate}[\indent(a)]
		\item All sector angles $a_1,a_2,...,a_N$ are congruent angles; $a_1=a_2=\cdots=a_N$. \label{remark:BC2a}
		\item When $N$ is an even number, all sector angles $a_1,a_2,...,a_N$ are the opposite angles of one another; $a_s=a_{\frac{N}{2}+s}$ for $s=1,...\frac{N}{2}$. \label{remark:BC2b}
	\end{enumerate}
        \label{remark:BC2}
    \end{remark}

    \begin{figure}[ht]
	\centering
	\subfloat[][$a_1=a_2$]
	{
		\centering\resizebox{0.22\textwidth}{!}{\includegraphics{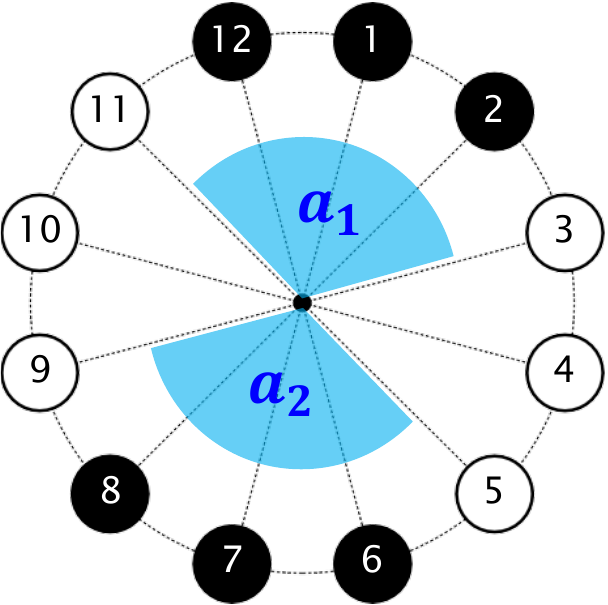}}
		\label{fig:ex_BC2_case1}		
	}
	~
	\subfloat[][$a_1=a_2=a_3$]
	{
		\centering\resizebox{0.22\textwidth}{!}{\includegraphics{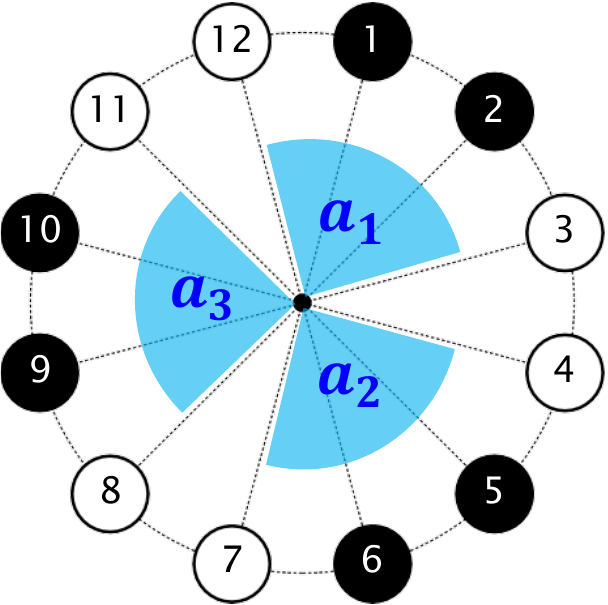}}
		\label{fig:ex_BC2_case2}		
	}
	~
	\subfloat[][$a_1=a_2=a_3=a_4$]
	{
		\centering\resizebox{0.22\textwidth}{!}{\includegraphics{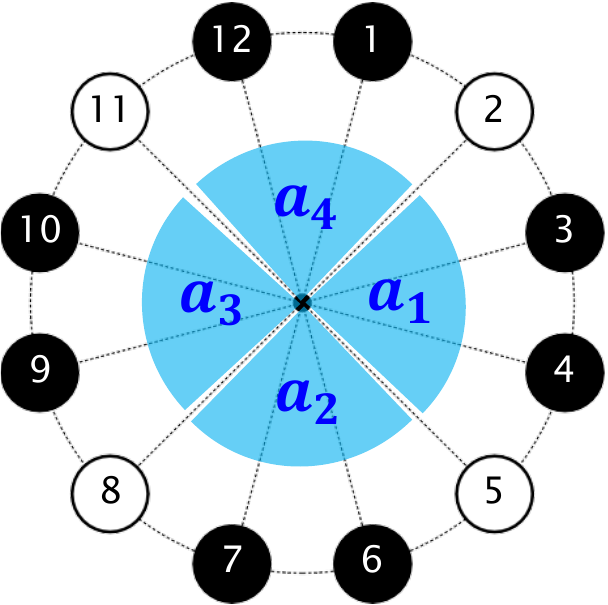}}
		\label{fig:ex_BC2_case3}		
	}
	~
	\subfloat[][$a_1=a_3$ and $a_2=a_4$]
	{
		\centering\resizebox{0.22\textwidth}{!}{\includegraphics{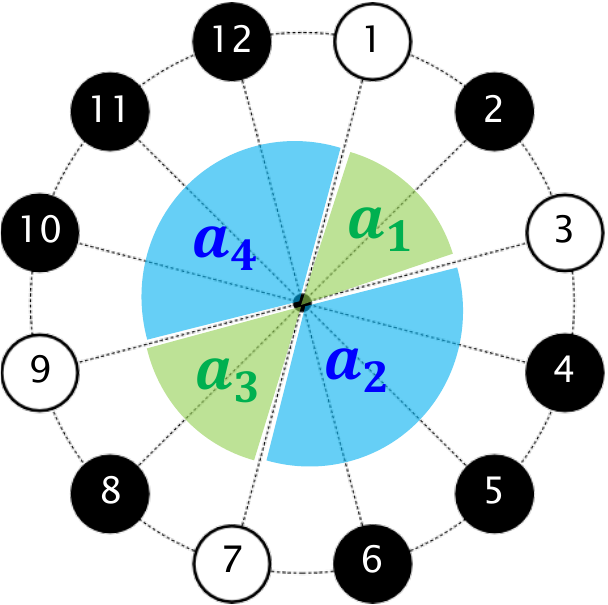}}
		\label{fig:ex_BC2_case4}		
	}
	\caption{Representative examples of the systems satisfying BC2}
	\label{fig:illustrative_example_BC2}
    \end{figure}
	
    Regarding the method of examining BC2 given a subsystem $U$, Endharta and Ko \cite{EK2020} previously found that the existence of more than one reverse distance tuple $E^{(l)}$ such that $E^{(l)}=D_U$ implies the spread proportionality of the system. Hence, the following remark can be used to check if a subsystem satisfies BC2.
    \begin{remark}[Endharta and Ko \cite{EK2020}] \label{remark:how_to_examing_bc2} 
	A \ckngb system with an index set of non-failed units $U$ satisfies BC2 if $|\mathbf{E}_U^B|>1$.
    \end{remark}
	
    Combining remark \ref{remark:how_to_examing_bc2} with remark \ref{remark:how_to_examing_bc1}, we draw a conclusion that a subsystem satisfying BC1 also satisfies BC2 whereas the opposite is not always true. For example, among the graphical models that depicted in Fig.~\ref{fig:illustrative_example_BC2}, only three of them also satisfies BC1 with at least a pair of perpendicular axes of symmetry (see Figs.~\ref{fig:ex_BC2_case1}, \ref{fig:ex_BC2_case3} and \ref{fig:ex_BC2_case4})). As such, BC2 is regarded as a more generalized balance condition than BC1. For another descriptive example, Fig.~\ref{fig:ex_BC2_not_BC1} shows a subsystem which is the same graphical model as previously shown in Fig.~\ref{fig:ex_BC1_3_axes}, which has been identified as unbalanced according to BC1. Now, we notice that it can be regarded as a balanced system by considering BC2. 
    \begin{figure}[ht]
	\centering
	\subfloat[][A subsystem satisfying BC2 but not BC1]
	{
		\centering\resizebox{0.31\textwidth}{!}{\includegraphics{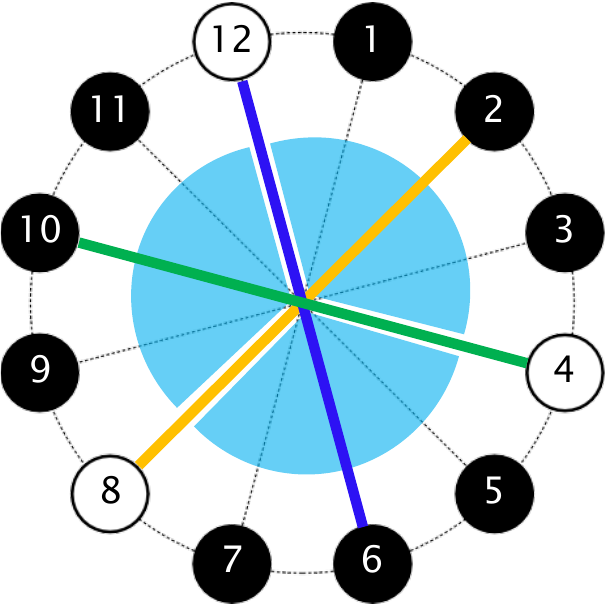}}
		\label{fig:ex_BC2_not_BC1}		
	}
	~
	\subfloat[][A subsystem not satisfying both BC1 and BC2]
	{
		\centering\resizebox{0.31\textwidth}{!}{\includegraphics{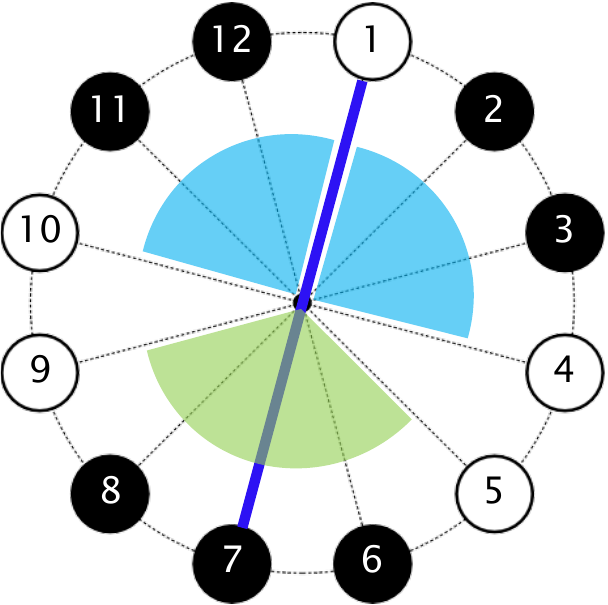}}
		\label{fig:ex_not_BC1_BC2}		
	}
	\caption{Illustrative examples for BC2}
	\label{fig:illustrative_example_BC2_with_BC1}
    \end{figure}
	
    In the next subsection, we will introduce the most generalized balance condition (to the best of our knowledge) BC3 that covers both BC1 and BC2. For example, although the subsystem in Fig.~\ref{fig:ex_not_BC1_BC2} cannot be identified as balanced even if we consider both BC1 and BC2, it will be regarded as balanced by taking BC3 into account (see Fig.~\ref{fig:ex_BC3_case1}).
	
    \subsection{BC3: \textit{System has a center of gravity at origin}} \label{subsec:bc3}
    The motivation of BC3 stems from the following research question: \textit{What is the essence that BC1 and BC2 are trying to evaluate?} In consequence, both BC1 and BC2 can be considered as the conditions representing the states in which a subsystem comprised of only a part of units can maintain the balance. Such condition may correspond to the states in which the center of gravity formed by the operating units is located at the geometric center of the system. To be specific, the \textit{center of gravity} stands for the mean point of an object's or system's weight distribution. By the assumption that we have homogeneous units, the system's direction of total thrust should be in the same direction as the lift force given the center of gravity located at the origin. Based on this simple and intuitive concept, we state the following definition for the new as well as more generalized balance condition.

    \begin{definition}[BC3]
        A \ckngb system is said to be balanced if the system with operating units has a center of gravity at the origin.
    \end{definition}
	
    For example, plots in Fig.~\ref{fig:illustrative_example_BC3} depict two graphical models describing the meaning of BC3. First, Fig.~\ref{fig:ex_BC3_case1} corresponds the case where BC3 is satisfied because the operating units in a subsystem $U=\{1,4,5,9,10\}$ forms the center of gravity at the origin; the blue-colored `x'-marker indicates the location of the center of gravity. Note also that the subsystem is the same as the previous figure Fig.~\ref{fig:ex_not_BC1_BC2} which was examined as unbalanced by considering both BC1 and BC2. On the other hand, Fig.~\ref{fig:ex_BC3_case2} shows another subsystem $U=\{3,6,8,12\}$ in which the center of gravity (the red-colored `x'-marker) is not located at origin, which is identified as unbalanced although the most generalized condition BC3 is considered.

    \begin{figure}[ht]
	\centering
	\subfloat[][A subsystem satisfying BC3 but not both BC1 and BC2]
	{
		\centering\resizebox{0.32\textwidth}{!}{\includegraphics{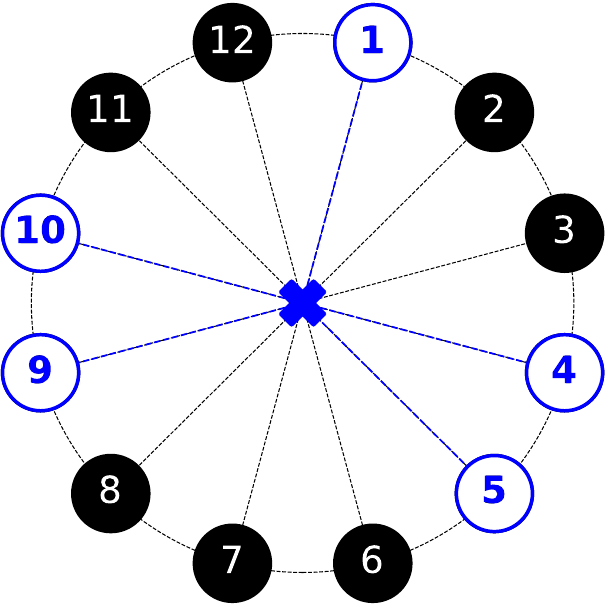}}
		\label{fig:ex_BC3_case1}		
	}
	~
	\subfloat[][A subsystem not satisfying BC1, BC2, and BC3]
	{
		\centering\resizebox{0.32\textwidth}{!}{\includegraphics{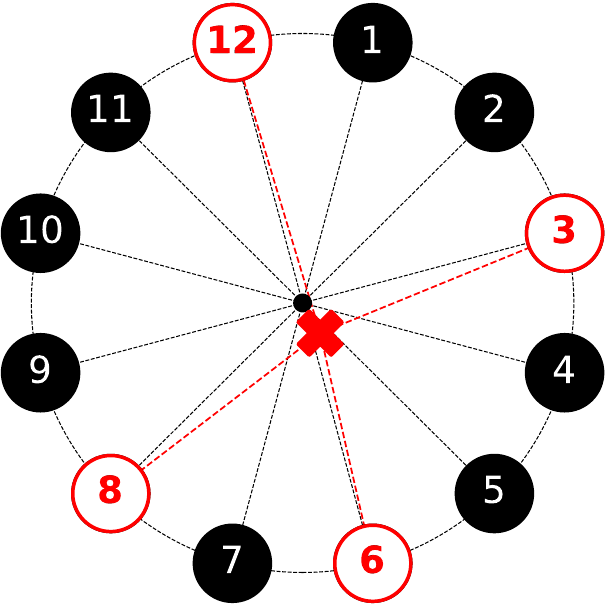}}
		\label{fig:ex_BC3_case2}		
	}
	\caption{Illustrative examples for BC3 (the colored `x'-markers represent the center of gravity for each subsystem)}
	\label{fig:illustrative_example_BC3}
    \end{figure}

	
    To quantitatively examine BC3 for a \ckngb system, we consider a conventional two-dimensional coordinate system with x-axis and y-axis. Without loss of generality, we assume that each unit of the system is located on a unit circle with radius 1. That is, the distance between the origin $(0,0)$ and each unit is $1$ for all units. Again, without loss of generality, we adjust the coordinate of unit $1$ to $(1,0)$ and arrange all the other units counterclockwise in ascending order of unit index such that the x-y coordinate of unit $i$, $(x_i, y_i)$, becomes $\left(\cos\left(\left(i-1\right)\theta\right),\ \sin\left(\left(i-1\right)\theta\right) \right)$ for $i=1,...,n$ where $\theta=2\pi/n$. Then, the coordinate of the center of gravity given a subsystem $U=\{u_1, u_2,...,u_{|U|}\}$ can be simply calculated by $\left( \left(\sum_{u\in U}x_u\right)/|U|, \left(\sum_{u\in U}y_u\right)/|U| \right)$. In summary, the following remark states a quantitative way of examining whether or not a subsystem $U$ satisfies BC3.
	
    \begin{remark} A \ckngb system with an index set of non-failed units $U$ satisfies BC3 if the following equality holds:
	\begin{align}
		\begin{bmatrix}
			0 \\
			0 
		\end{bmatrix}
		= 
		\begin{bmatrix}
			\displaystyle \frac{\sum_{u\in U} \cos\left( \left( u-1 \right)\theta  \right)}{|U|} \\
			\displaystyle \frac{\sum_{u\in U} \sin\left( \left( u-1 \right)\theta  \right)}{|U|}
		\end{bmatrix}
		.
	\end{align}
    \end{remark}
	
    For example, consider the graphical model in Fig.~\ref{fig:ex_BC3_case1} showing a subsystem $U=\{1,4,5,9,10\}$ of a circular $k$-out-of-$12$: G balanced system. Noting that $\theta=2\pi/12=\pi/6$ and $|U|=5$, the center of gravity, say $(\bar{x}_U,\bar{y}_U)$, can be obtained as follows
    \begin{align*}
	\begin{bmatrix}
		\bar{x}_U \\
		\bar{y}_U
	\end{bmatrix}       
	& \equiv 
	\begin{bmatrix}
		\displaystyle \frac{\sum_{u\in U} \cos\left( \left( u-1 \right)\theta  \right)}{|U|} \\
		\displaystyle \frac{\sum_{u\in U} \sin\left( \left( u-1 \right)\theta  \right)}{|U|}
	\end{bmatrix} 
	\\
	& = 
	\begin{bmatrix}
		\displaystyle \frac{\cos\left(0\cdot\frac{\pi}{6}\right)+\cos\left(3\cdot\frac{\pi}{6}\right)+\cos\left(4\cdot\frac{\pi}{6}\right)+\cos\left(8\cdot\frac{\pi}{6}\right)+\cos\left(9\cdot\frac{\pi}{6}\right)}{5} \\
		\displaystyle \frac{\sin\left(0\cdot\frac{\pi}{6}\right)+\sin\left(3\cdot\frac{\pi}{6}\right)+\sin\left(4\cdot\frac{\pi}{6}\right)+\sin\left(8\cdot\frac{\pi}{6}\right)+\sin\left(9\cdot\frac{\pi}{6}\right)}{5}
	\end{bmatrix} \\
	& =
	\begin{bmatrix}
		0 \\
		0
	\end{bmatrix}
	,
    \end{align*}
    which verifies that a subsystem $U=\{1,4,5,9,10\}$ satisfies BC3 hence is identified as balanced.
	
    For another example, consider the graphical model in  Fig.~\ref{fig:ex_BC3_case2} showing a subsystem $U=\{3,6,8,12\}$ with $\theta=\pi/6$ and $|U|=4$. Similarly, the center of gravity can be obtained as follows
    \begin{align*}
	\begin{bmatrix}
		\bar{x}_U \\
		\bar{y}_U
	\end{bmatrix}       
	& = 
	\begin{bmatrix}
		\displaystyle \frac{\cos\left(2\cdot\frac{\pi}{6}\right)+\cos\left(5\cdot\frac{\pi}{6}\right)+\cos\left(7\cdot\frac{\pi}{6}\right)+\cos\left(11\cdot\frac{\pi}{6}\right)}{4} \\
		\displaystyle \frac{\sin\left(2\cdot\frac{\pi}{6}\right)+\sin\left(5\cdot\frac{\pi}{6}\right)+\sin\left(7\cdot\frac{\pi}{6}\right)+\sin\left(11\cdot\frac{\pi}{6}\right)}{4}
	\end{bmatrix} \\
	& \approx
	\begin{bmatrix}
		-0.0915 \\
		0.0915
	\end{bmatrix}
	\neq
	\begin{bmatrix}
		0 \\
		0
	\end{bmatrix}
	,
    \end{align*}
    which concludes that a subsystem $U=\{3,6,8,12\}$ does not satisfy BC3 hence is unbalanced.
	
    \subsection{Relationship between the balance conditions} \label{subsec:relationship}
    In this subsection, we discuss the relationship among the three balance conditions BC1, BC2, and BC3. The main purpose of this investigation is to show that BC3 is the most generalized balance condition. First, we start with the following proposition which shows that BC3 is a more generalized balance condition than BC1.
    \begin{proposition}\label{prop:b1b3}
	BC1 implies BC3.
    \end{proposition}
    \begin{proof}
	See Appendix \ref{app:A}.
    \end{proof}
	
    Next, the following proposition states that BC3 is a more generalized balance condition that BC2. Recall that we have observed an example of this relationship graphically by comparing Fig.~\ref{fig:ex_not_BC1_BC2} with Fig.~\ref{fig:ex_BC3_case1}.
    \begin{proposition}\label{prop:b2b3}
	BC2 implies BC3.
    \end{proposition}
    \begin{proof}
	See Appendix \ref{app:B}.
    \end{proof}
	
    Simply combining the propositions \ref{prop:b1b3} and \ref{prop:b2b3}, we have the following conclusion.
    \begin{corollary}\label{cor:b3b1}
	BC3 is the most generalized balance condition among BC1, BC2, and BC3.
    \end{corollary}
    \begin{proof}
	By propositions \ref{prop:b1b3} and \ref{prop:b2b3}.
    \end{proof}
	
    In the next section, we explain how to evaluate the system reliability based on the well-known minimum tie-set method considering the balance conditions discussed so far.
	
    \section{Reliability Evaluation}\label{sec:reliability}
    The reliability of \ckngb systems can be evaluated considering the balance conditions. We apply the minimum tie-set method (also known as minimal path set method) \cite{ELSAYED2021}. That is, we interpret the system as a parallel system consists of minimum tie-sets and find the probability that at least one minimum tie-set is operational. In the following two subsections, We will introduce the notion of minimum tie-sets in the context of \ckngb systems and then explain how to evaluate the system reliability. 
	
    \subsection{Minimum tie-sets}
    A \textit{tie-set} (also known as \textit{path set}) is a subset of units in the system which by operating ensures the system is functioning. Hence, a tie-set for a \ckngb system should include at least $k$ elements and satisfy one of the balance conditions. Among the ordinary tie-sets, the \textit{minimum tie-sets} (also known as \textit{minimal path sets}) are the tie-sets that cannot be reduced without losing their property as a tie-set. The minimum tie-set plays an important role when we evaluate the system reliability since the system can be interpreted as a parallel structure of minimum tie-sets \cite{ELSAYED2021}. 
	
    To examine the minimality of a tie-set, we need to check its inclusion relationship with other tie-sets; a tie-set which is a superset of another tie-set cannot be minimal. The flowchart in Fig.~\ref{fig:flowchart_minimum_tie_sets_enumeration} summarizes the procedure to find all the minimum tie-sets in a \ckngb system. Note that the procedure applies for all the balance definitions by changing the balance condition checking step.
    \begin{figure}
	\centering   
	\includegraphics[width=0.85\linewidth]{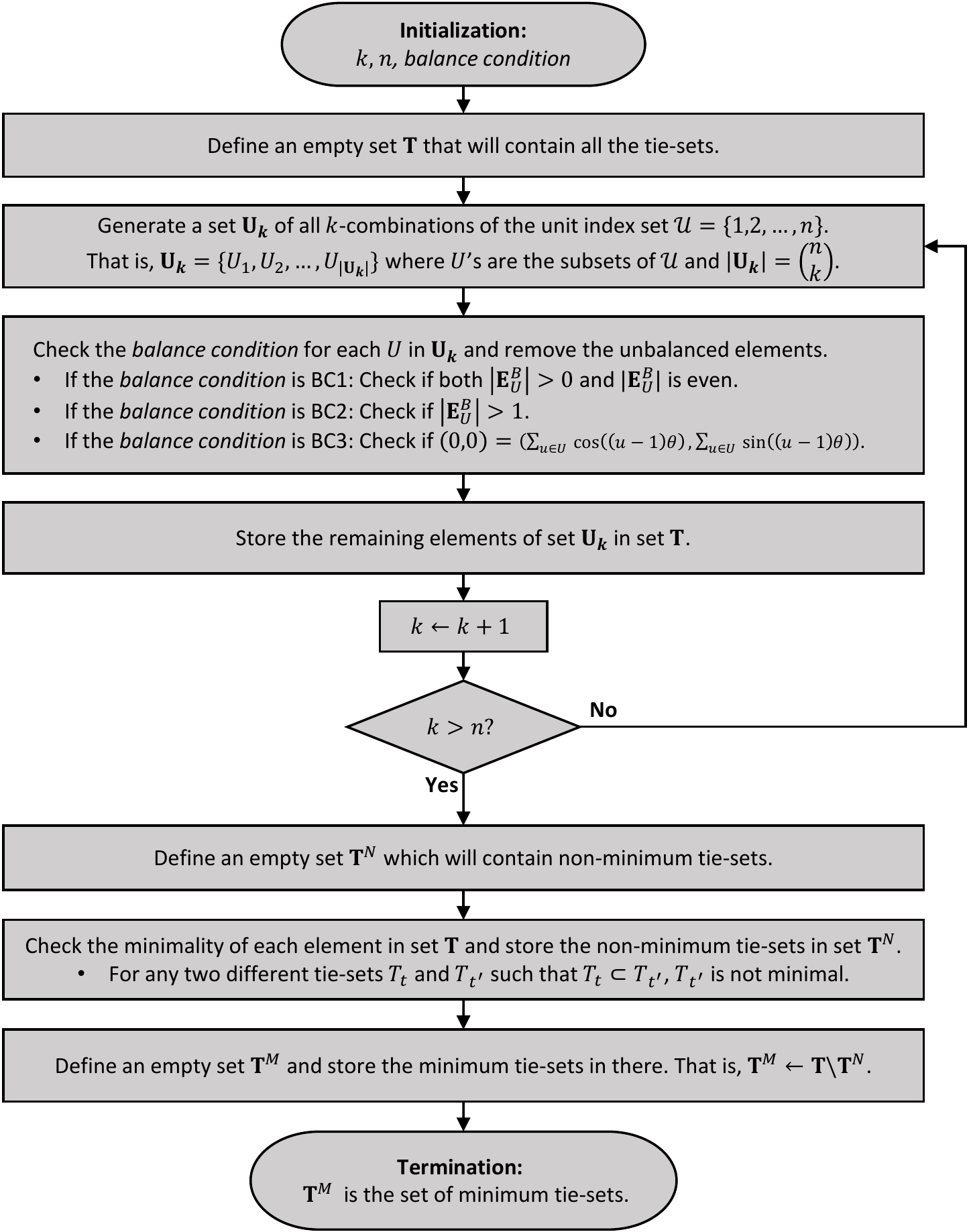}  
	\caption{Procedure for enumerating all the minimum tie-sets} 
	\label{fig:flowchart_minimum_tie_sets_enumeration}
    \end{figure}

    \subsection{System reliability}
    In this subsection, we explain how to evaluate the system reliability using the set of minimum tie-sets. Let $X_i$ be the binary state variable of unit $i$; $X_i=1$ if unit $i$ is functioning, $X_i=0$ otherwise. Collecting the states of all units, we let $\mathbf{X}$ be the system state vector; $\mathbf{X}=[X_1,X_2,...,X_n]$. Then, we define the system structure function $\phi(\cdot)$ of $\mathbf{X}$ as follows:
	\begin{align}
	\phi(\mathbf{X}) \equiv 1-\prod_{T^M \in \mathbf{T}^M}\left( 1-\prod_{u\in T^M} X_{u} \right). \nonumber
	\end{align}
	By the above definition, $\phi(\mathbf{X})$ becomes an indicator random variable that is 1 if $\mathbf{X}$ corresponds to an operational system and 0 otherwise. 
	
    Using the system structure function, the system reliability $R_S$ can be directly derived by calculating $\mathbb{P}\left[\phi(\mathbf{X})=1\right]$ which is equal to $\mathbb{E}\left[\phi(\mathbf{X})\right]$ due to the property of indicator random variable. Since the probability that a unit is functioning properly is a constant $r$ and all units are assumed to be homogeneous, we have $\mathbb{E}[X_i]=r$ for all $i$. Together with the assumption of independent units, $R_S$ is obtained by
    \begin{align}
	R_S = \mathbb{E}\left[\phi(\mathbf{X})\right] 
	& = \mathbb{E}\left[ 1-\prod_{T_M \in \mathbf{T}^M}\left( 1-\prod_{u\in T^M} X_{u} \right) \right] \quad\textrm{(by the definition of reliability)} \nonumber \\
	& = 1-\prod_{T^M \in \mathbf{T}^M}\left( 1-\prod_{u\in T^M} \mathbb{E}\left[ X_{u} \right] \right) \quad\textrm{(by independence)} \nonumber \\
	& = 1-\prod_{T^M \in \mathbf{T}^M}\left( 1-\prod_{u\in T^M} r \right)\quad\textrm{(by homogeneity)} \nonumber  \\
	& = 1-\prod_{T^M \in \mathbf{T}^M}\left( 1-r^{|T^M|} \right). \label{eq:system_reliability}
    \end{align}    

    From the last Eq. (\ref{eq:system_reliability}) for the system reliability, calculating the value of $R_S$ is straightforward once the set of minimum tie-sets $\mathbf{T}^M$ is acquired by enumeration. Next, we provide a specific example of system reliability evaluation for descriptive purpose.
	
    \subsection{A descriptive numerical example}
    Consider a circular $4$-out-of-$12$: G balanced system. Following the procedures in Fig.~\ref{fig:flowchart_minimum_tie_sets_enumeration}, we enumerate all the minimum-tie sets according to each balance definition and visualize them in Fig.~\ref{fig:MTS_graphical_enumeration}. As shown in the figure, we have 15 minimum tie-sets considering only BC1, 19 minimum tie-sets considering BC2, and 31 minimum tie-sets considering BC3, for a circular $4$-out-of-$12$: G balanced system. Since the Eq. (\ref{eq:system_reliability}) implies that the system reliability $R_S$ is proportional to the number of minimum tie-sets $|\mathbf{T}^M|$, the increment of the number of minimum tie-sets will result in the enhancement of the system reliability. To better understand the systems with different parameters, Table \ref{table:num_MTS_comparison} compares the number of minimum tie-sets for various combinations of $(k,n)$ for each balance condition. Although not all the systems show the difference in the number of minimum tie-sets between BC3 and BC2, we observe significant differences (larger than 10) for some systems such as $4$-out-of-$12$ ($\textrm{diff}=12$), $6$-out-of-$12$ ($\textrm{diff}=21$), $6$-out-of-$14$ ($\textrm{diff}=14$), and $8$-out-of-$14$ ($\textrm{diff}=14$). 
	
    \begin{figure}
	\centering   
	\includegraphics[width=0.9\linewidth]{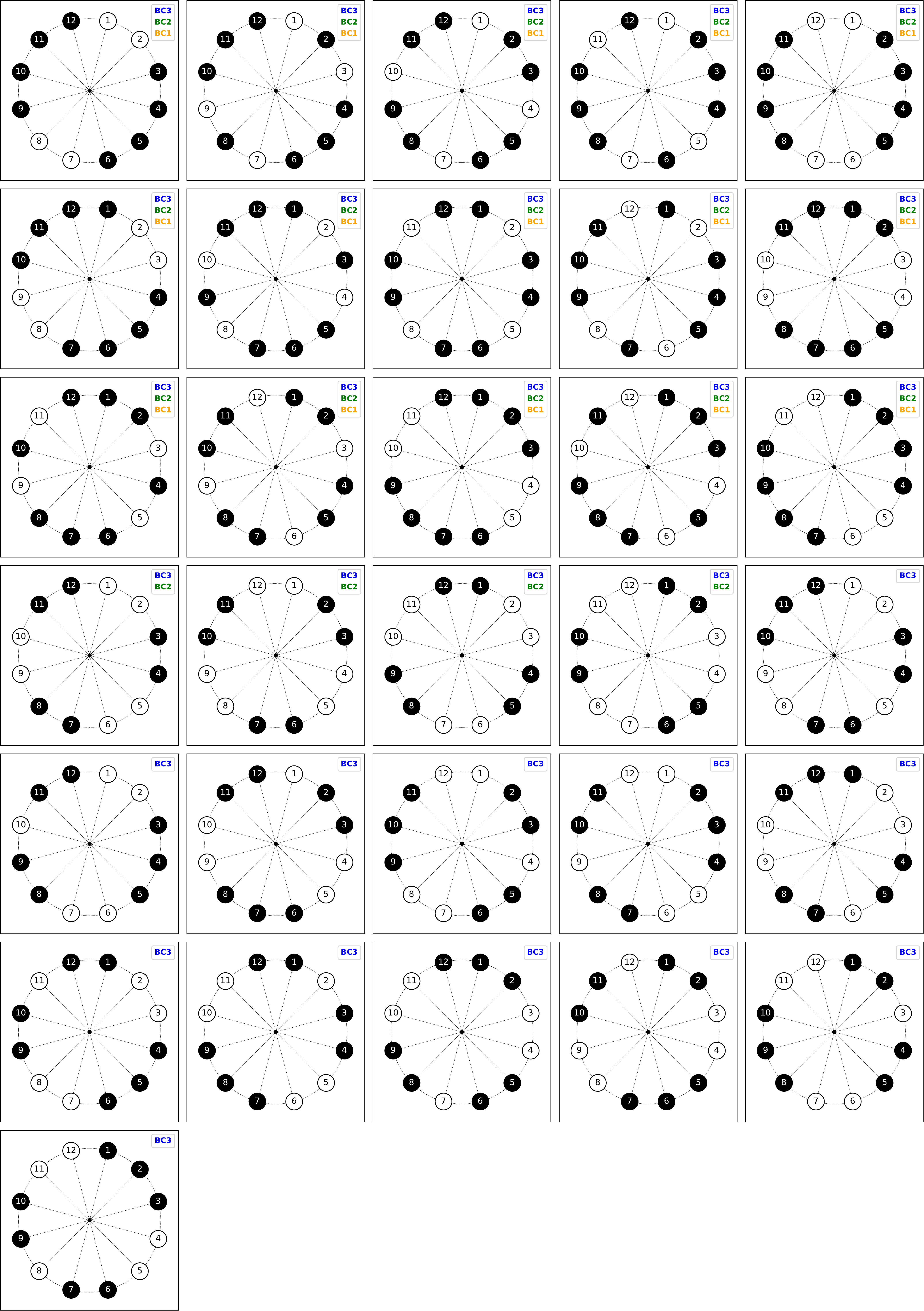}  
	\caption{Enumeration of all the graphical models of minimum tie-sets in a circular $4$-out-of-$12$: G balanced system (each legend exhibits the balance conditions that the corresponding minimum tie-set satisfies)} 
	\label{fig:MTS_graphical_enumeration}
    \end{figure}   	

	In the next section, we provide extensive numerical studies on the reliability enhancement effect resulted from introducing the new balance condition BC3.

    \begin{table}[H]
	\centering
	\caption{Comparison of the number of minimum tie-sets considering BC1, BC2, and BC3}
	\label{table:num_MTS_comparison}
	\resizebox{0.9\textwidth}{!}{%
		\begin{tblr}{
				colspec = {Q[108]Q[108]Q[117]Q[117]Q[117]Q[187]Q[187]},
				cells = {c},
				cell{1}{1} = {c=2}{0.2\linewidth},
				cell{1}{3} = {c=3}{0.351\linewidth},
				cell{1}{6} = {c=2}{0.216\linewidth},
				cell{3}{2} = {r=2}{},
				cell{5}{2} = {r=3}{},
				cell{8}{2} = {r=4}{},
				cell{12}{2} = {r=5}{},
				cell{17}{2} = {r=6}{},
				vlines,
				hline{1-3,5,8,12,17,23} = {-}{},
				hline{4,6-7,9-11,13-16,18-22} = {1,3-7}{},
			}
			\textbf{System Parameters} &  & \textbf{Number of Minimum Tie-sets} &  &  & \textbf{Difference} & \\
			$k$ & $n$ & BC1 & BC2 & BC3 & BC2$-$BC1 & \textbf{BC3$-$BC2}\\
			2 & 6 & 3 & 5 & 5 & 2 & 0\\
			4 &  & 3 & 3 & 3 & 0 & 0\\
			2 & 8 & 4 & 4 & 4 & 0 & 0\\
			4 &  & 6 & 6 & 6 & 0 & 0\\
			6 &  & 4 & 4 & 4 & 0 & 0\\
			2 & 10 & 5 & 7 & 7 & 2 & 0\\
			4 &  & 10 & 12 & 12 & 2 & 0\\
			6 &  & 10 & 10 & 10 & 0 & 0\\
			8 &  & 5 & 5 & 5 & 0 & 0\\
			2 & 12 & 6 & 10 & 10 & 4 & 0\\
			4 &  & 15 & 19 & 31 & 4 & \textbf{\underline{12}}\\
			6 &  & 11 & 15 & 36 & 4 & \textbf{\underline{21}}\\
			8 &  & 15 & 15 & 19 & 0 & \textbf{\underline{4}}\\
			10 &  & 6 & 6 & 6 & 0 & 0\\
			2 & 14 & 7 & 9 & 9 & 2 & 0\\
			4 &  & 21 & 23 & 23 & 2 & 0\\
			6 &  & 21 & 23 & 37 & 2 & \textbf{\underline{14}}\\
			8 &  & 21 & 21 & 35 & 0 & \textbf{\underline{14}}\\
			10 &  & 21 & 21 & 21 & 0 & 0\\
			12 &  & 7 & 7 & 7 & 0 & 0
		\end{tblr}
	}
    \end{table}
	
    \section{Numerical Analysis} \label{sec:numerical}
    We investigate the reliability of \ckngb systems with various system parameters $k$ (minimum number of non-failed units for working system), $n$ (number of units in the system), and $r$ (unit reliability), for different balance condition consideration. Though the main purpose of this section is to show the reliability enhancement effect thanks to the new balance definition BC3 compared to BC2, we include the BC1's results together in order to provide richer experimental data that helps understanding the gradual improvement from BC1 to BC2 to BC3.
	
    \subsection{Reliability comparison for the systems with various parameters}
    Fig.~\ref{fig:reliability_comparison_type_1} depicts the effect of the unit reliability on the system reliability for each balance condition (say $R_S^{\textrm{BC1}}$, $R_S^{\textrm{BC2}}$, and $R_S^{\textrm{BC3}}$) for $k$-out-of-12 (three plots in upper row) and $k$-out-of-14 (another three plots in lower row) systems where $k=4,6,8$. The first important but anticipatory observation is that considering more generalized balance condition consistently improves the system reliability for all the systems: $R_S^{\textrm{BC3}}\ge R_S^{\textrm{BC2}}\ge R_S^{\textrm{BC1}}$ for any $r$ in $[0,1]$. For some systems such as $6$-out-of-$12$, the level of improvement is quite significant so that introducing the new balance condition seems to be promising in terms of operating \ckngb systems. Second, the system reliability shows definite ``S''-shaped curve regardless of the system parameters. Together with the first observation, it is  implied that the effect of unit-wise reliability enhancement on the system reliability is stronger in the middle of $r$'s range $[0,1]$ and goes weaker as it approaches extreme values. 
    \begin{figure}[H]
	\centering
	\subfloat[][$4$-out-of-$12$]
	{
		\centering\resizebox{0.31\textwidth}{!}{\includegraphics{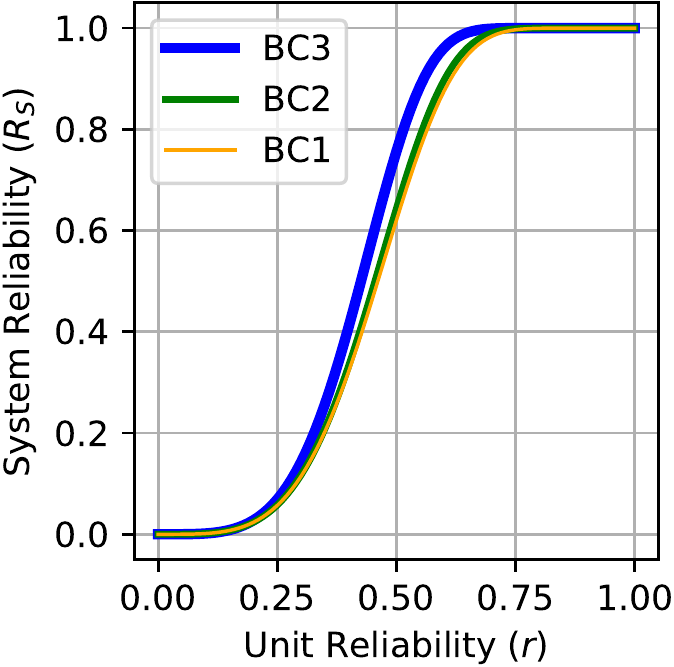}}
		\label{fig:type1_k_4_n_12}		
	}
	~
	\subfloat[][$6$-out-of-$12$]
	{
		\centering\resizebox{0.31\textwidth}{!}{\includegraphics{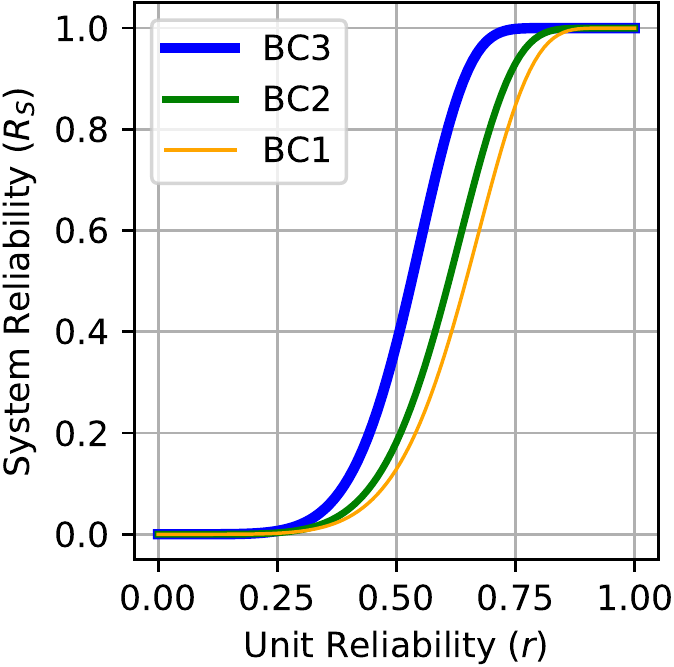}}
		\label{fig:type1_k_6_n_12}		
	}
	~
	\subfloat[][$8$-out-of-$12$]
	{
		\centering\resizebox{0.31\textwidth}{!}{\includegraphics{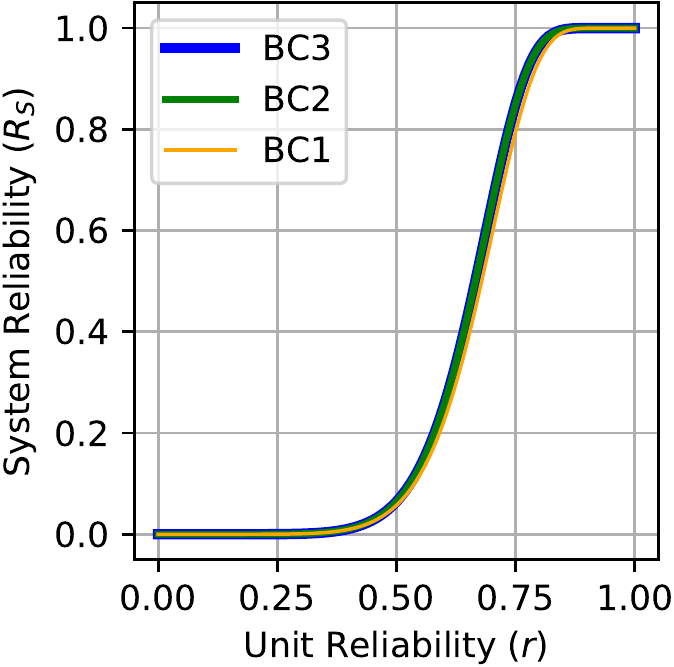}}
		\label{fig:type1_k_8_n_12}		
	}
	
	\subfloat[][$4$-out-of-$14$]
	{
		\centering\resizebox{0.31\textwidth}{!}{\includegraphics{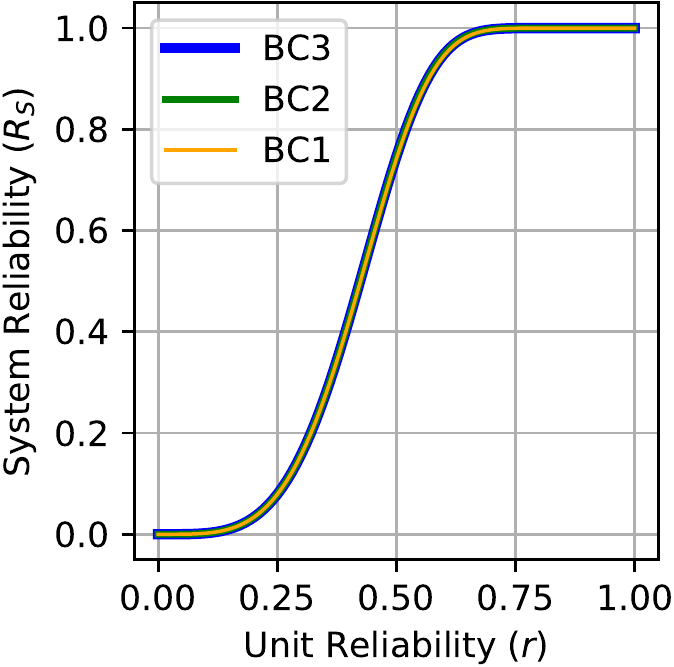}}
		\label{fig:type1_k_4_n_14}		
	}
	~
	\subfloat[][$6$-out-of-$14$]
	{
		\centering\resizebox{0.31\textwidth}{!}{\includegraphics{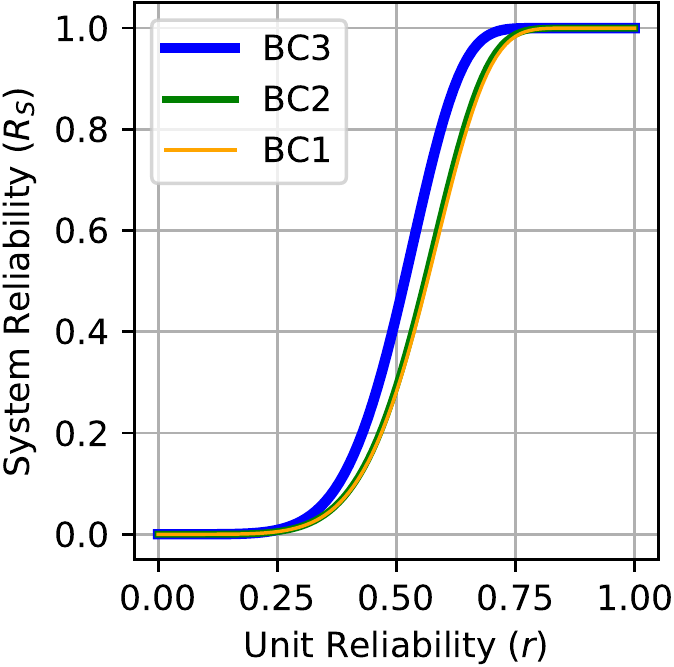}}
		\label{fig:type1_k_6_n_14}		
	}
	~
	\subfloat[][$8$-out-of-$14$]
	{
		\centering\resizebox{0.31\textwidth}{!}{\includegraphics{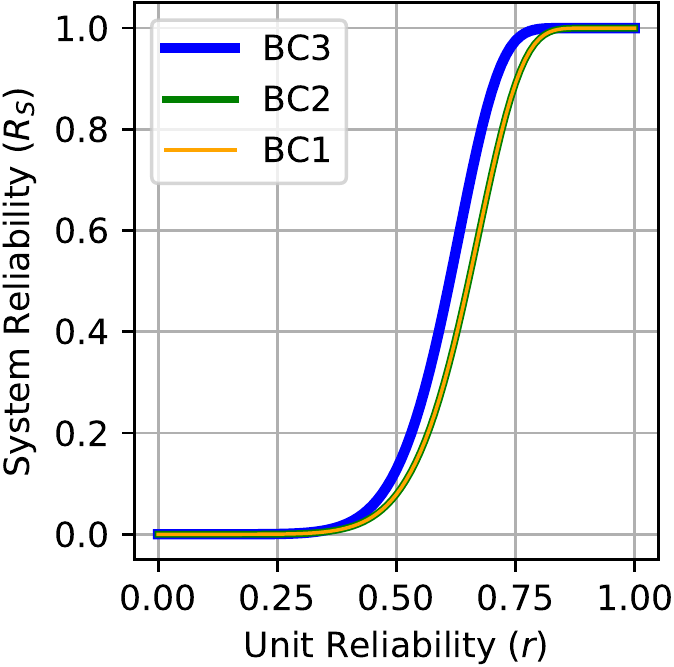}}
		\label{fig:type1_k_8_n_14}		
	}
	
	\caption{System reliability plots with varying unit reliability $r$ for the three balance conditions (upper row: $k$-out-of-12, lower row: $k$-out-of-14)}
	\label{fig:reliability_comparison_type_1}
    \end{figure}
	
    Fig.~\ref{fig:reliability_comparison_type_2} shows the system reliability with varying $k$'s and $r$'s for the systems with $n=12$ (upper row) and $n=14$ (lower row) under each balance condition. Similar to the observations from the previous plots, the system reliability shows ``S''-shaped curve in terms of unit reliability $r$ with the increasing trends from BC1 to BC2 to BC3, for all the cases. In addition, we find the decreasing tendency of $R_S$ as $k$ increases because the systems with larger $k$ require more non-failed units for the system to be operational. Note that the degree of decrement looks more significant when $k$ increases from an even number to an odd number (say $k_o$), for example, see the changing slopes of $R_S$ for $k=4,5,6$ for each plot in Figs.~\ref{fig:type2_n_14_BC1}--\ref{fig:type2_n_14_BC3}. This phenomenon aligns with our intuition because maintaining the physical balance with the odd number of units must be more difficult than with the even number of units. Hence, there must be only few additional minimum tie-sets for the system with $k=k_o$ compared to the system with $k=k_o + 1$ (the smallest even number greater than $k_o$).
    \begin{figure}[H]
	\centering
	\subfloat[][$n=12$, BC1]
	{
		\centering\resizebox{0.31\textwidth}{!}{\includegraphics{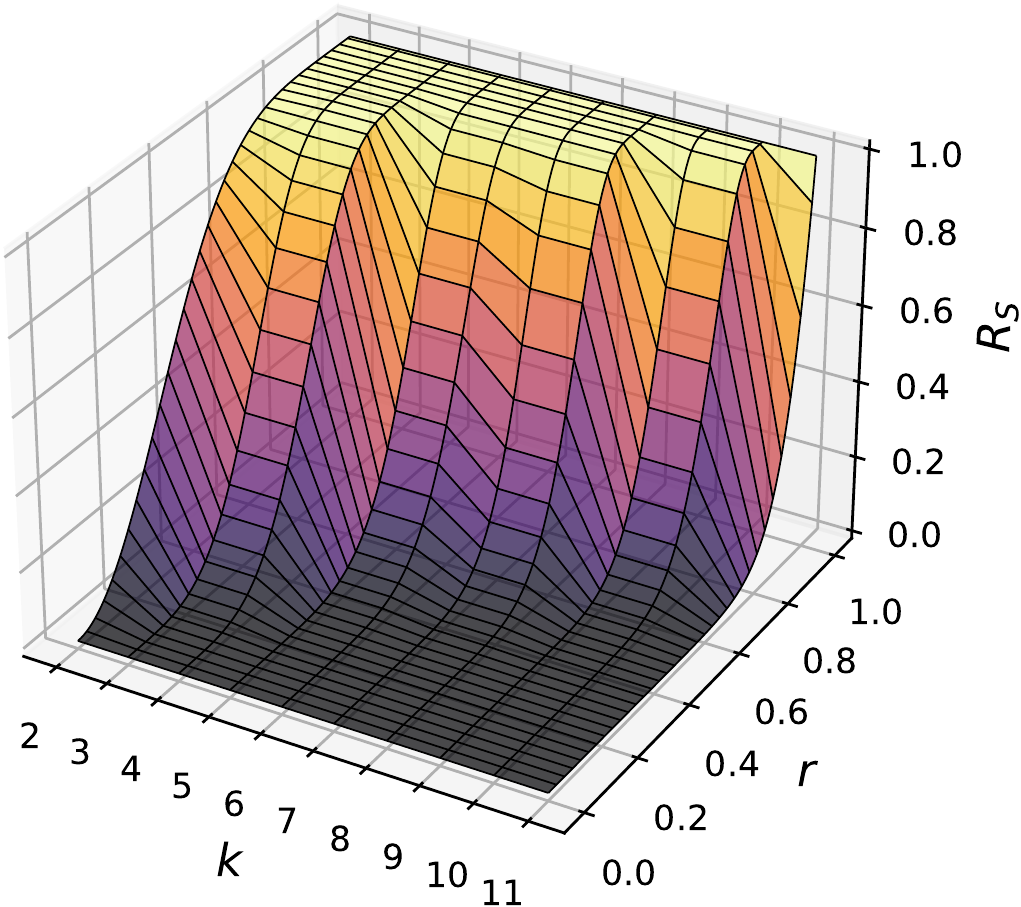}}
		\label{fig:type2_n_12_BC1}		
	}
	~
	\subfloat[][$n=12$, BC2]
	{
		\centering\resizebox{0.31\textwidth}{!}{\includegraphics{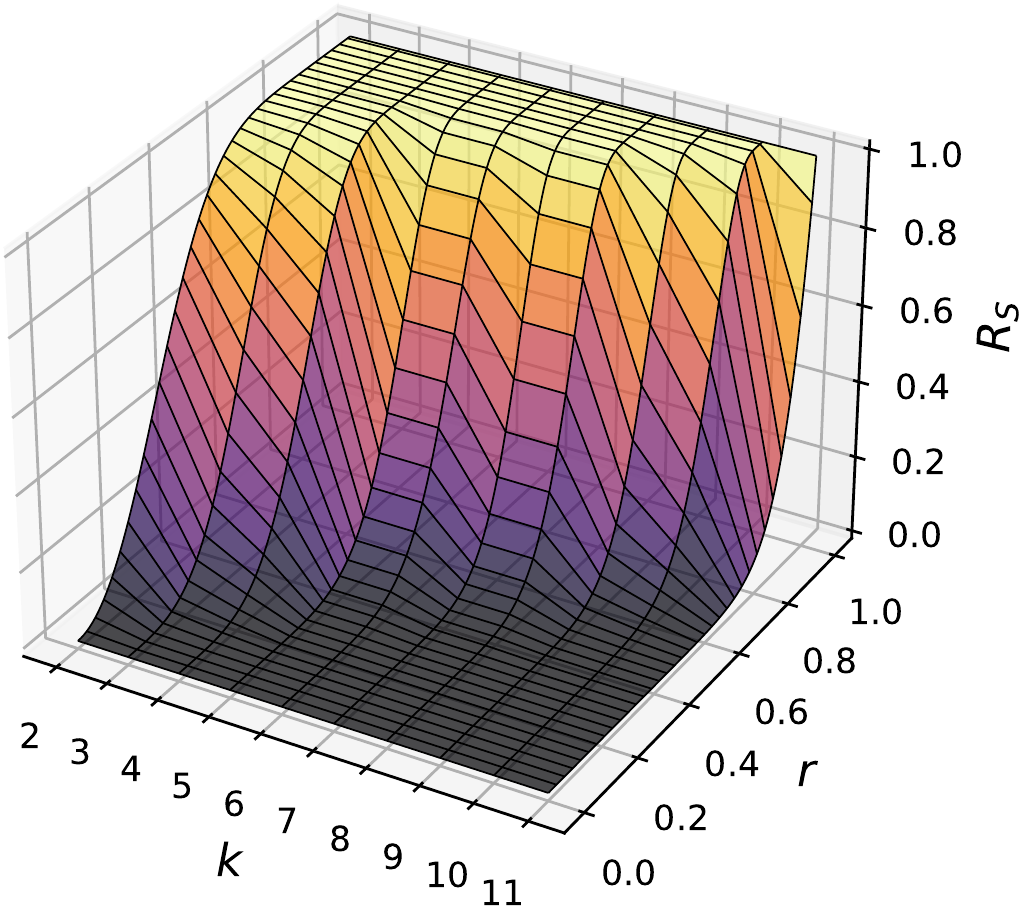}}
		\label{fig:type2_n_12_BC2}		
	}
	~
	\subfloat[][$n=12$, BC3]
	{
		\centering\resizebox{0.31\textwidth}{!}{\includegraphics{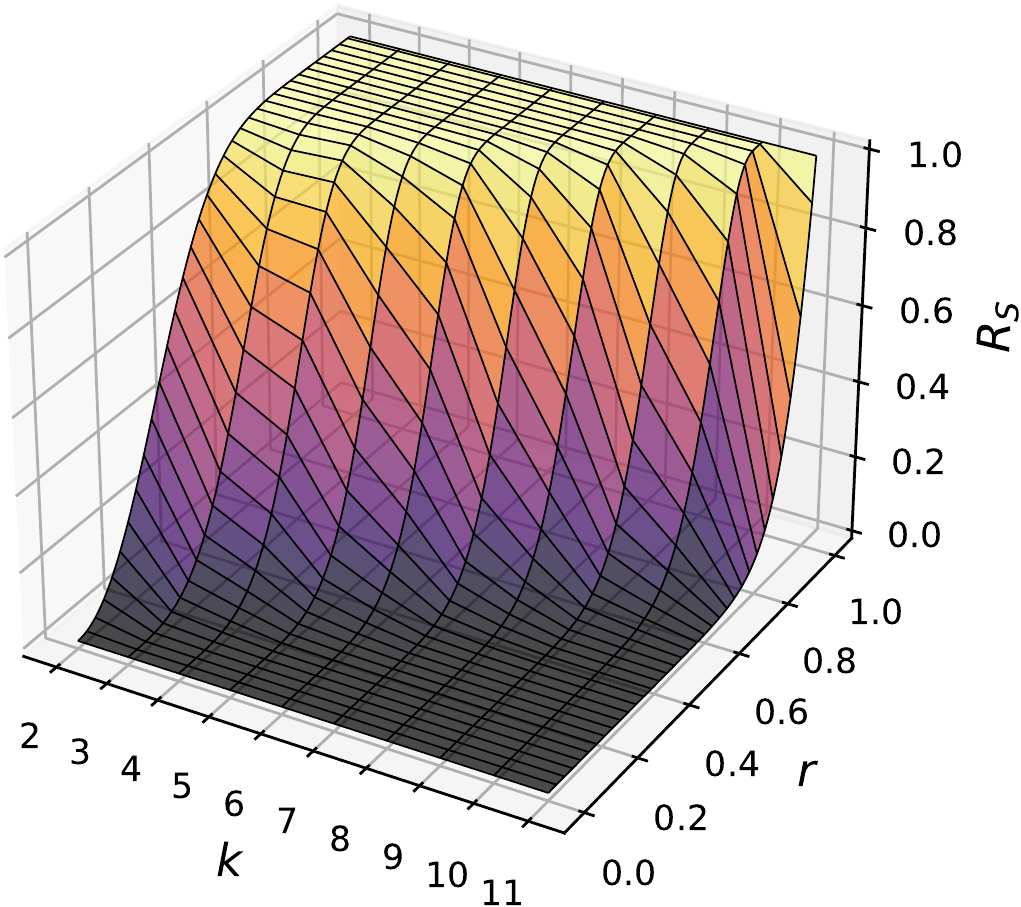}}
		\label{fig:type2_n_12_BC3}		
	}
	
	\subfloat[][$n=14$, BC1]
	{
		\centering\resizebox{0.31\textwidth}{!}{\includegraphics{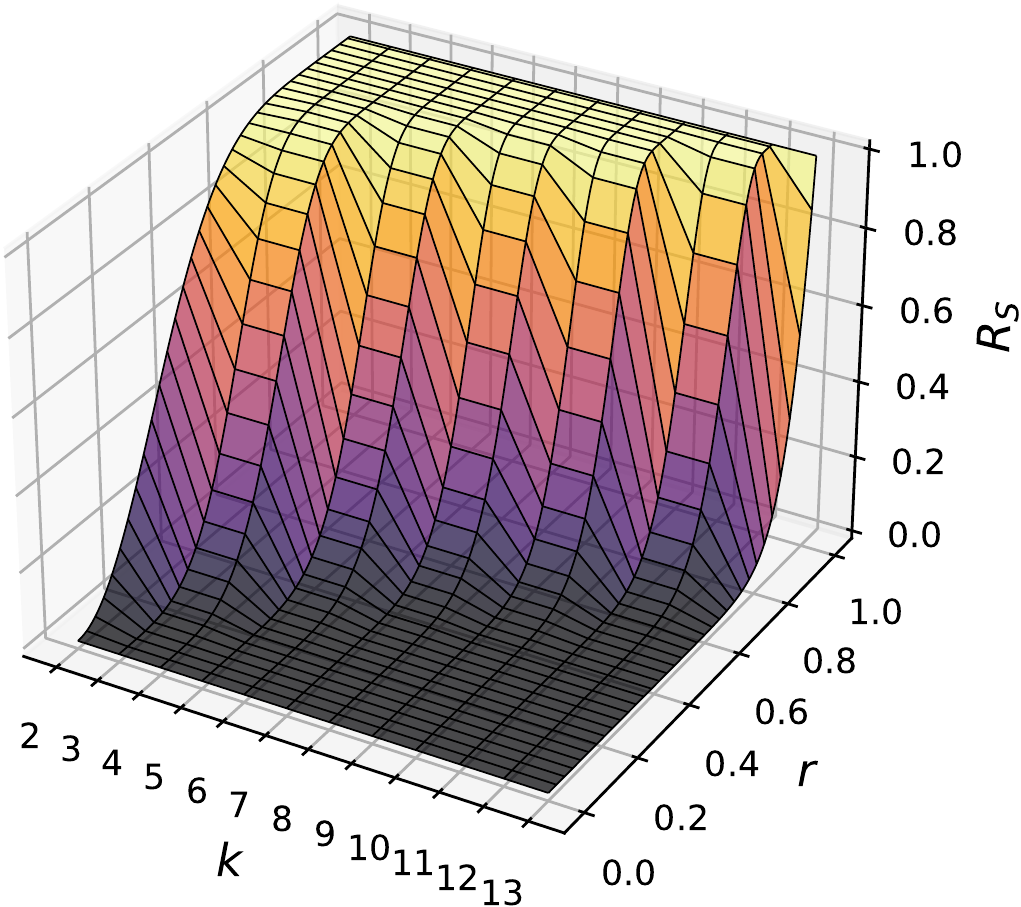}}
		\label{fig:type2_n_14_BC1}		
	}
	~
	\subfloat[][$n=14$, BC2]
	{
		\centering\resizebox{0.31\textwidth}{!}{\includegraphics{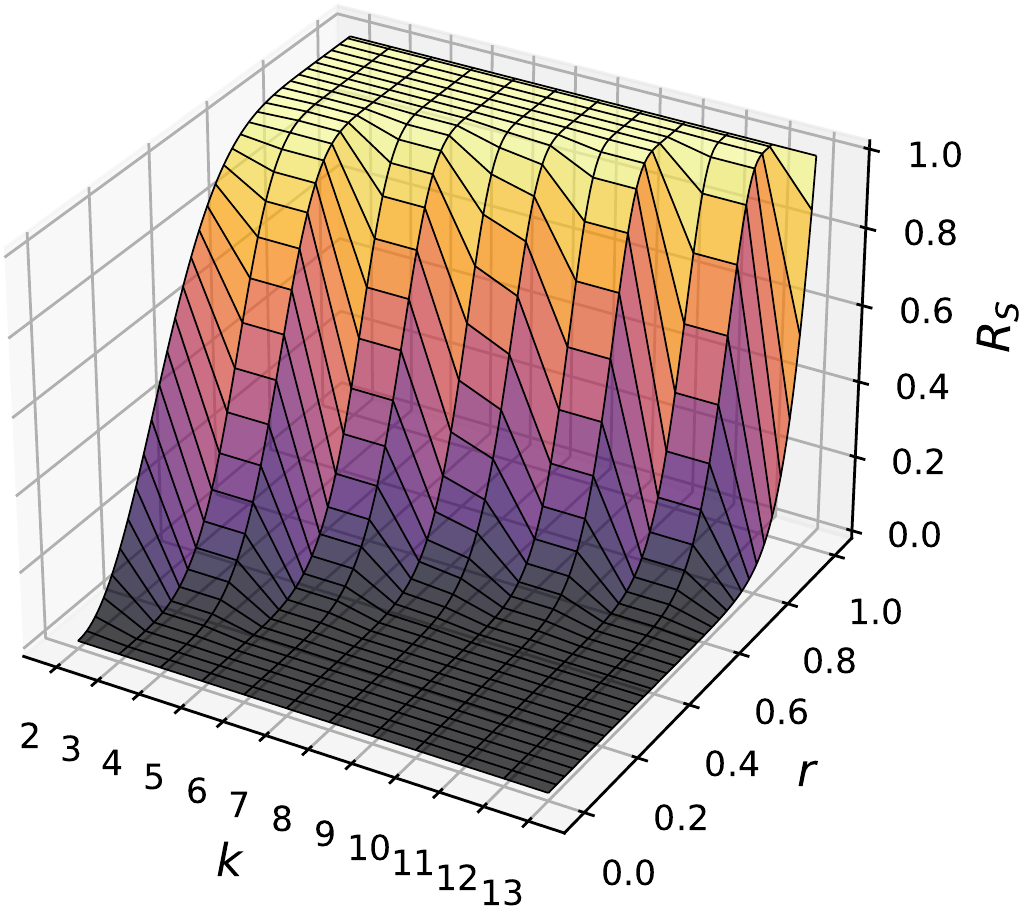}}
		\label{fig:type2_n_14_BC2}		
	}
	~
	\subfloat[][$n=14$, BC3]
	{
		\centering\resizebox{0.31\textwidth}{!}{\includegraphics{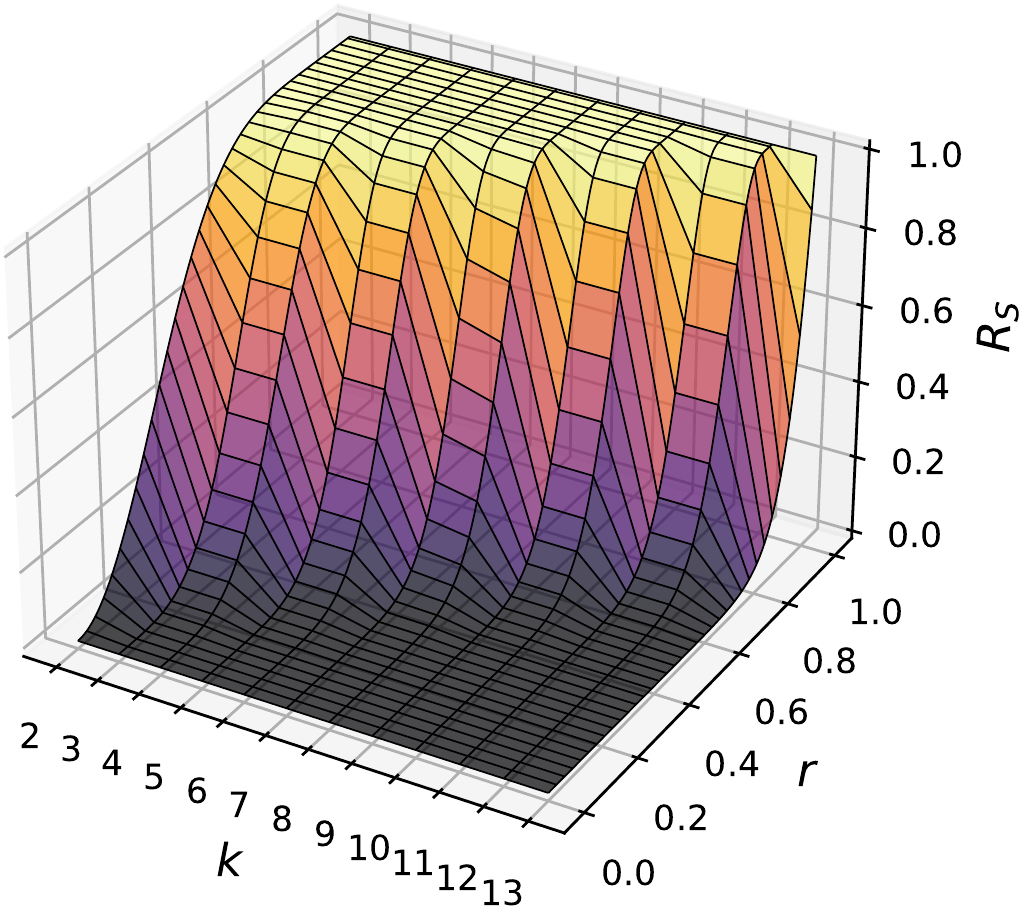}}
		\label{fig:type2_n_14_BC3}		
	}
	
	\caption{System reliability plots with varying $k$ and $r$ for the systems with $n=12$ (upper row) and $n=14$ (lower row) under each balance condition}
	\label{fig:reliability_comparison_type_2}
    \end{figure}
	
    Fig.~\ref{fig:reliability_comparison_type_3} provides the system reliability plots with varying $k$ and $n$ for $r=0.5,0.7,0.9$ and different balance conditions. It is worth mentioning that the reliability is only evaluated for the systems with $k$ and $n$ such that $k\le n-1$ hence the plots have empty space for some pairs of $(k,n)$. From the plots in Fig.~\ref{fig:reliability_comparison_type_3}, we understand the overall changing shape of the system reliability $R_S$ affected by changing $k$, $n$, and $r$. For example, we observe the decreasing trend of $R_S$ as $k$ increases for a fixed $n$ for all the cases. Similar to the observed phenomena in Fig.~\ref{fig:reliability_comparison_type_2}, the level of $R_S$'s decrement looks significant when $k$ is increased from an even number to an odd number. On the other hand, $R_S$ shows increasing trend as $n$ increases for a fixed $k$ for all the cases. The sensitivity of $R_S$ to $(k,n)$ looks greater for smaller $r$'s (see Figs.~\ref{fig:type3_r_0.5_BC1}--\ref{fig:type3_r_0.7_BC3}) whereas we mostly observe few changes except for the marginal area of the plots when $r$ is larger (see Figs.~\ref{fig:type3_r_0.9_BC1}--\ref{fig:type3_r_0.9_BC3}).

    \begin{figure}[H]
	\centering
	\subfloat[][$r=0.5$, BC1]
	{
		\centering\resizebox{0.31\textwidth}{!}{\includegraphics{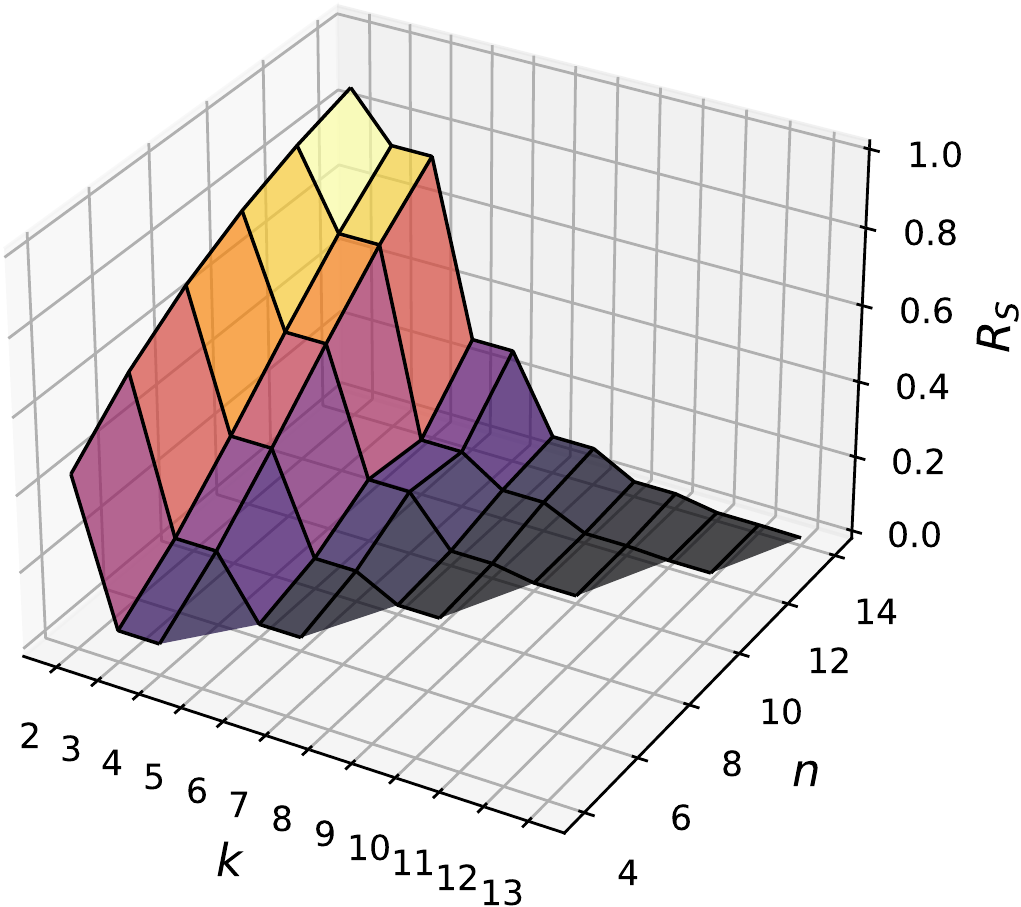}}
		\label{fig:type3_r_0.5_BC1}		
	}
	~
	\subfloat[][$r=0.5$, BC2]
	{
		\centering\resizebox{0.31\textwidth}{!}{\includegraphics{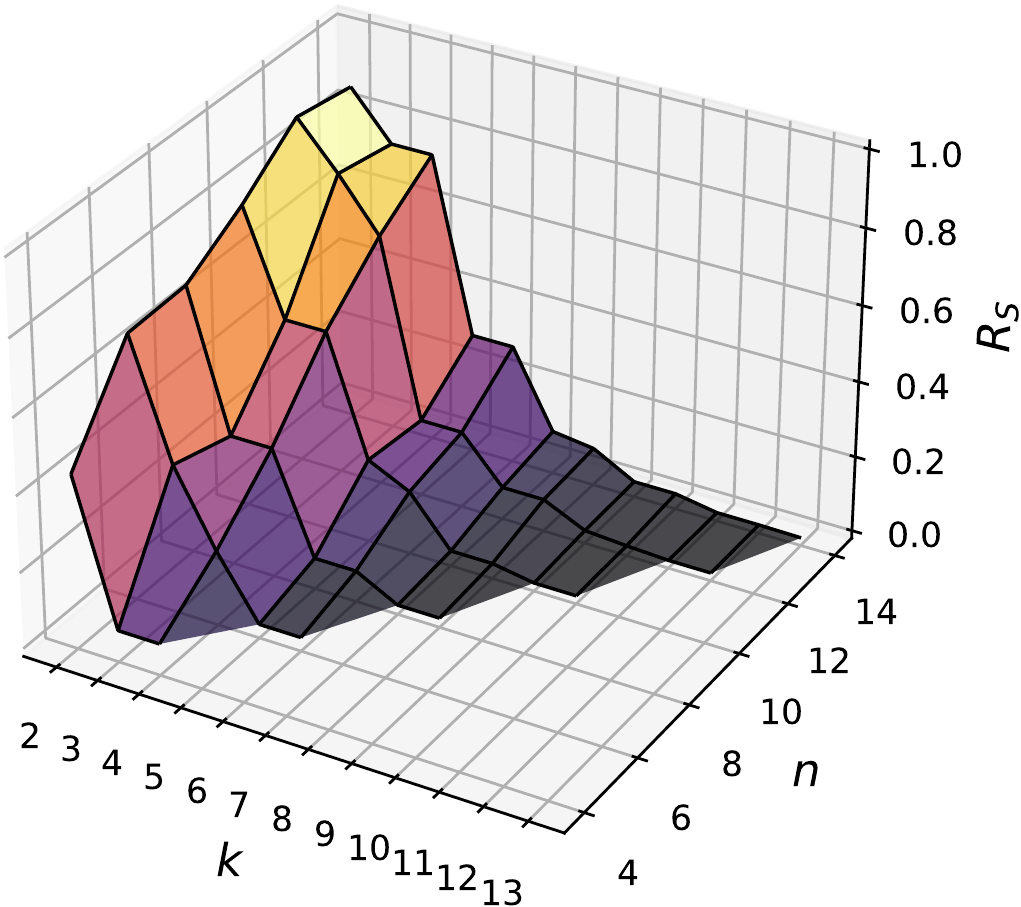}}
		\label{fig:type3_r_0.5_BC2}		
	}
	~
	\subfloat[][$r=0.5$, BC3]
	{
		\centering\resizebox{0.31\textwidth}{!}{\includegraphics{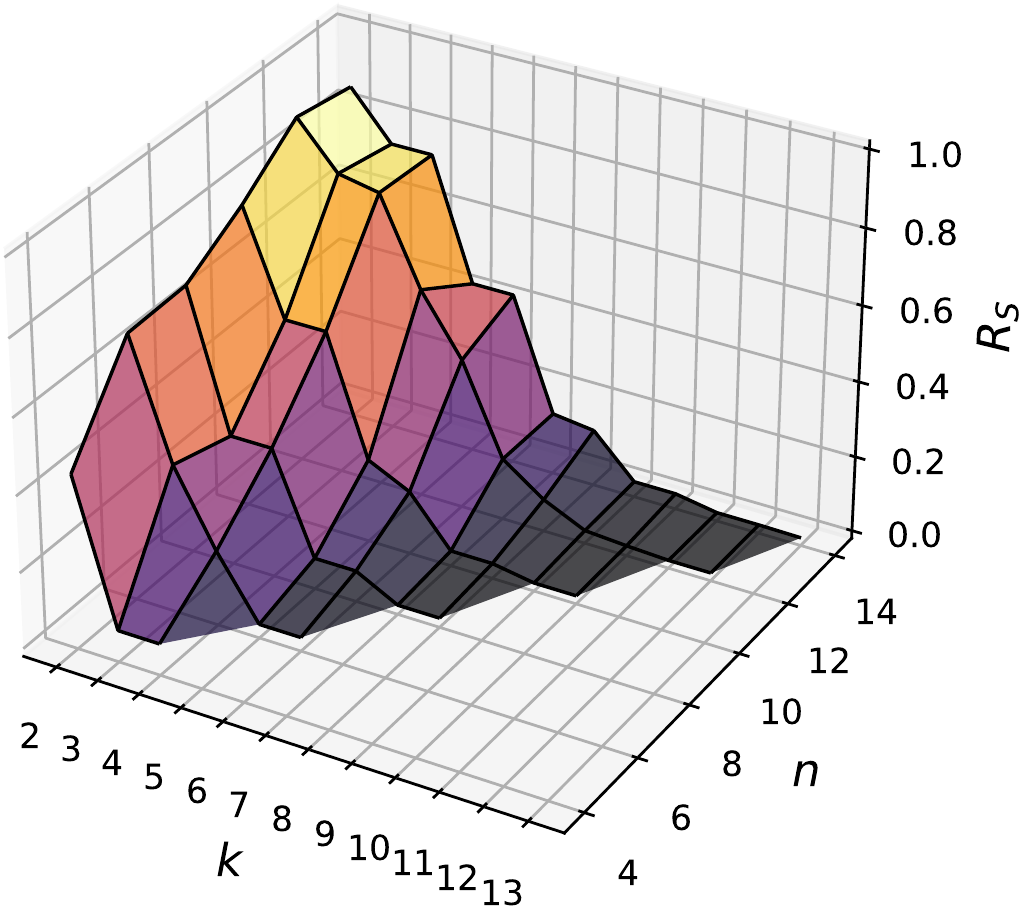}}
		\label{fig:type3_r_0.5_BC3}		
	}
	
	\subfloat[][$r=0.7$, BC1]
	{
		\centering\resizebox{0.31\textwidth}{!}{\includegraphics{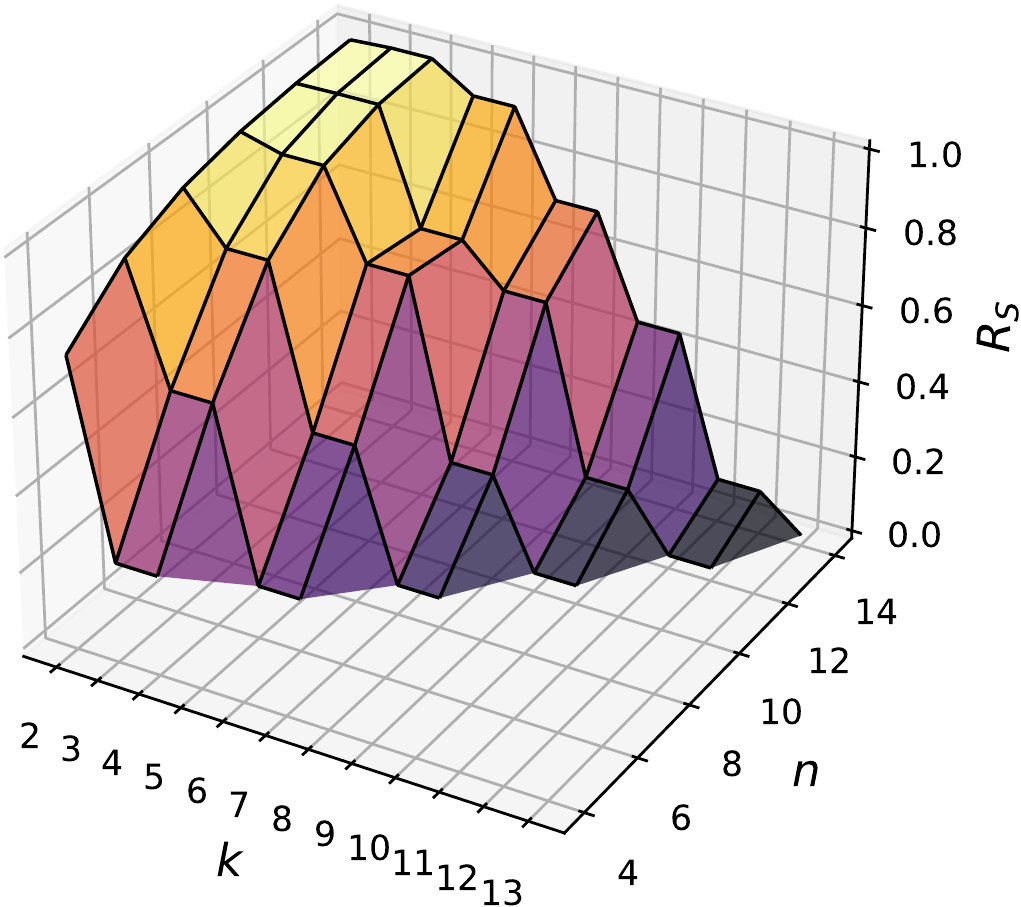}}
		\label{fig:type3_r_0.7_BC1}		
	}
	~
	\subfloat[][$r=0.7$, BC2]
	{
		\centering\resizebox{0.31\textwidth}{!}{\includegraphics{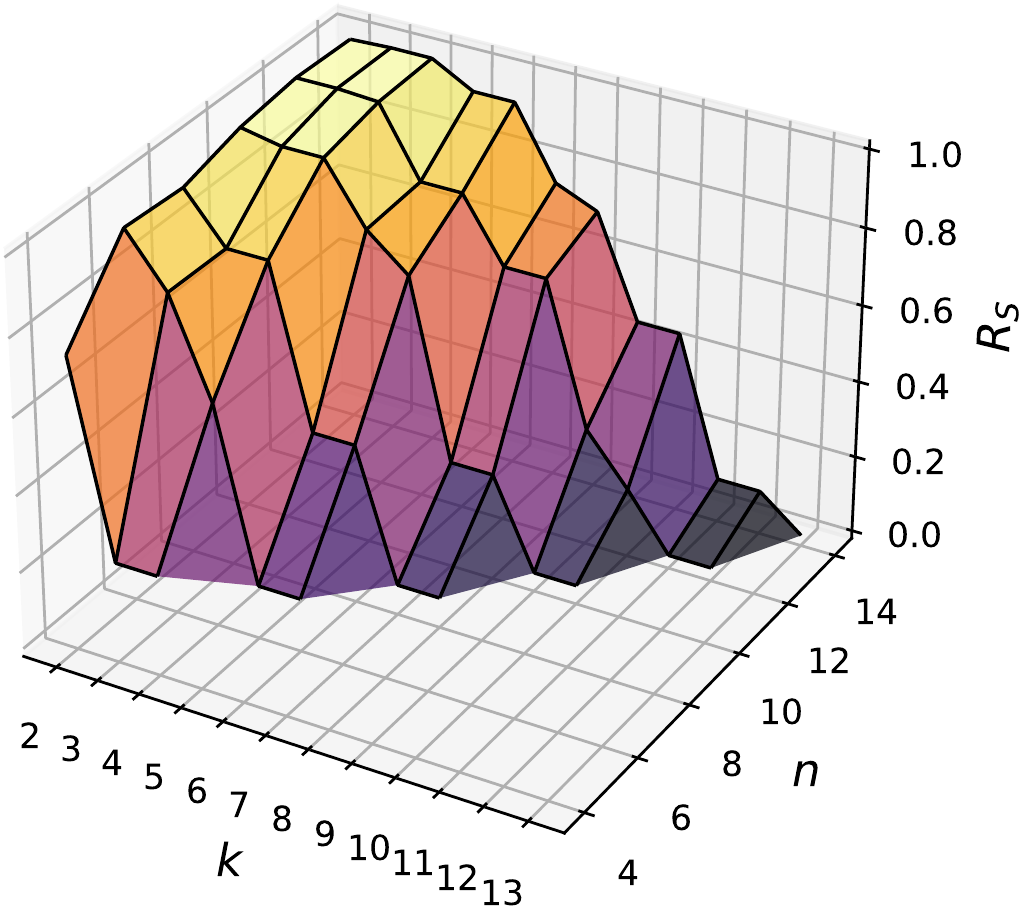}}
		\label{fig:type3_r_0.7_BC2}		
	}
	~
	\subfloat[][$r=0.7$, BC3]
	{
		\centering\resizebox{0.31\textwidth}{!}{\includegraphics{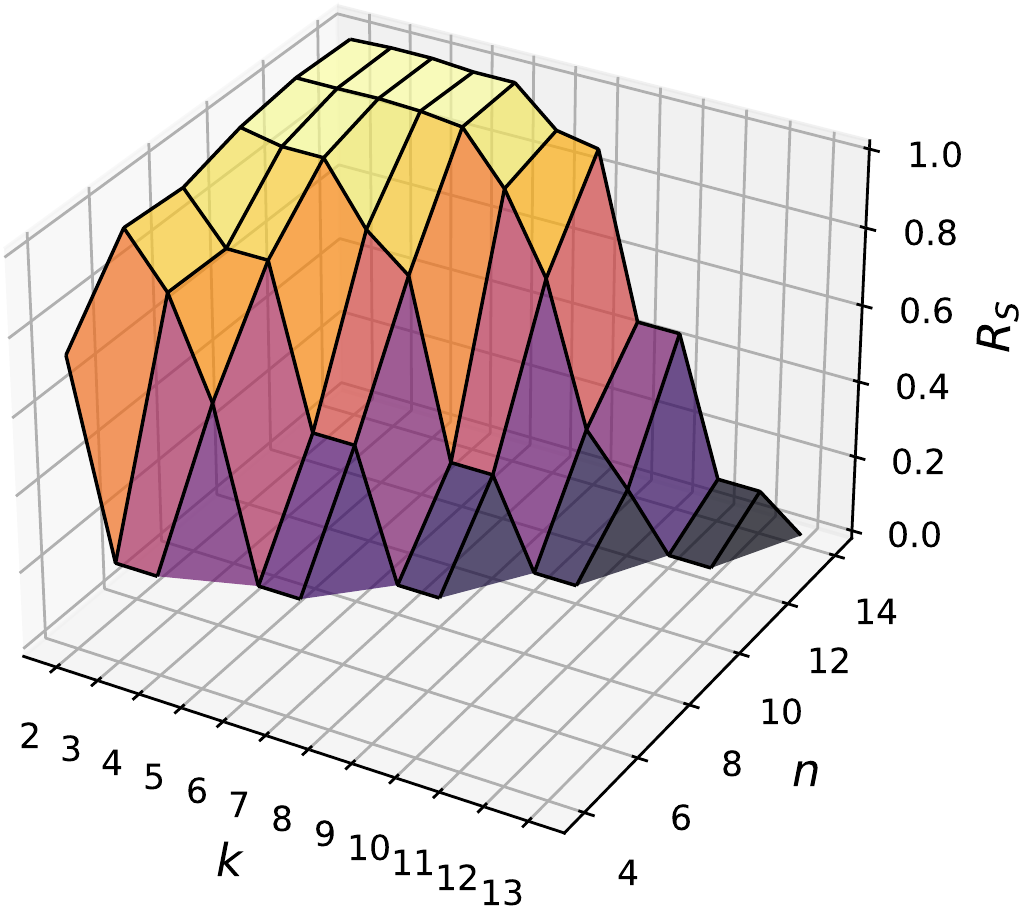}}
		\label{fig:type3_r_0.7_BC3}		
	}
	
	\subfloat[][$r=0.9$, BC1]
	{
		\centering\resizebox{0.31\textwidth}{!}{\includegraphics{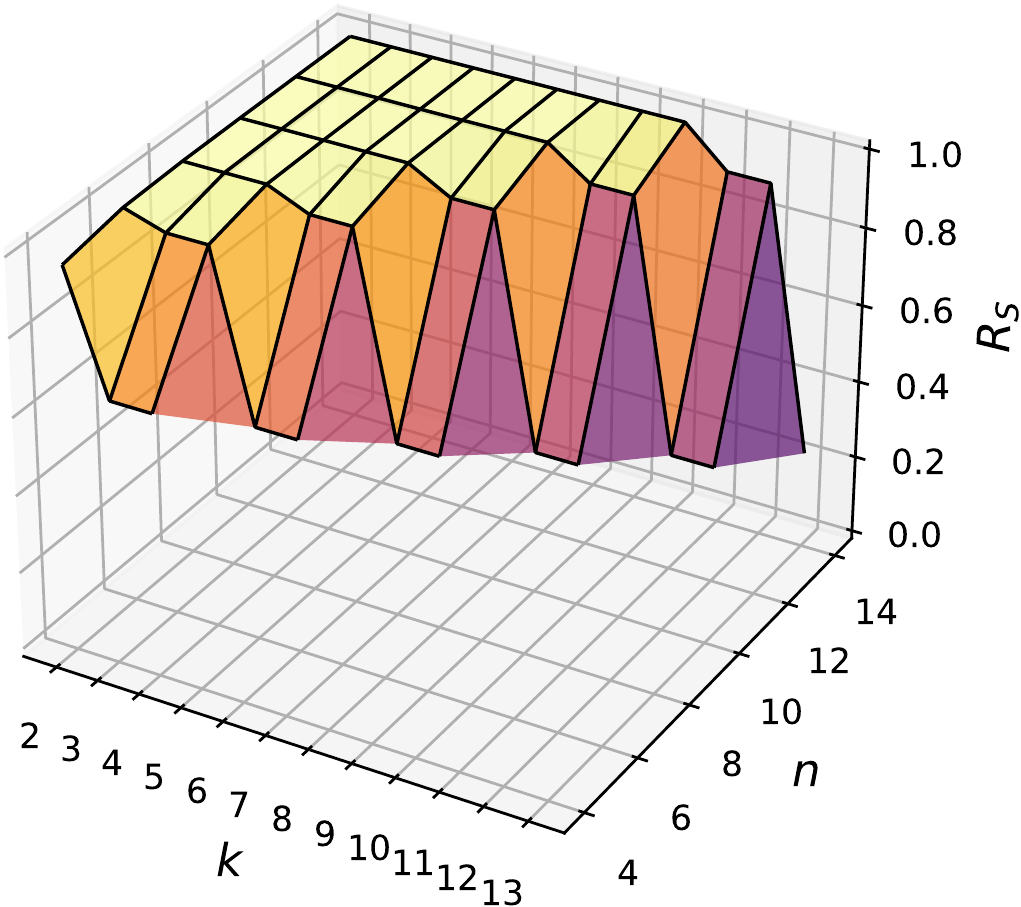}}
		\label{fig:type3_r_0.9_BC1}		
	}
	~
	\subfloat[][$r=0.9$, BC2]
	{
		\centering\resizebox{0.31\textwidth}{!}{\includegraphics{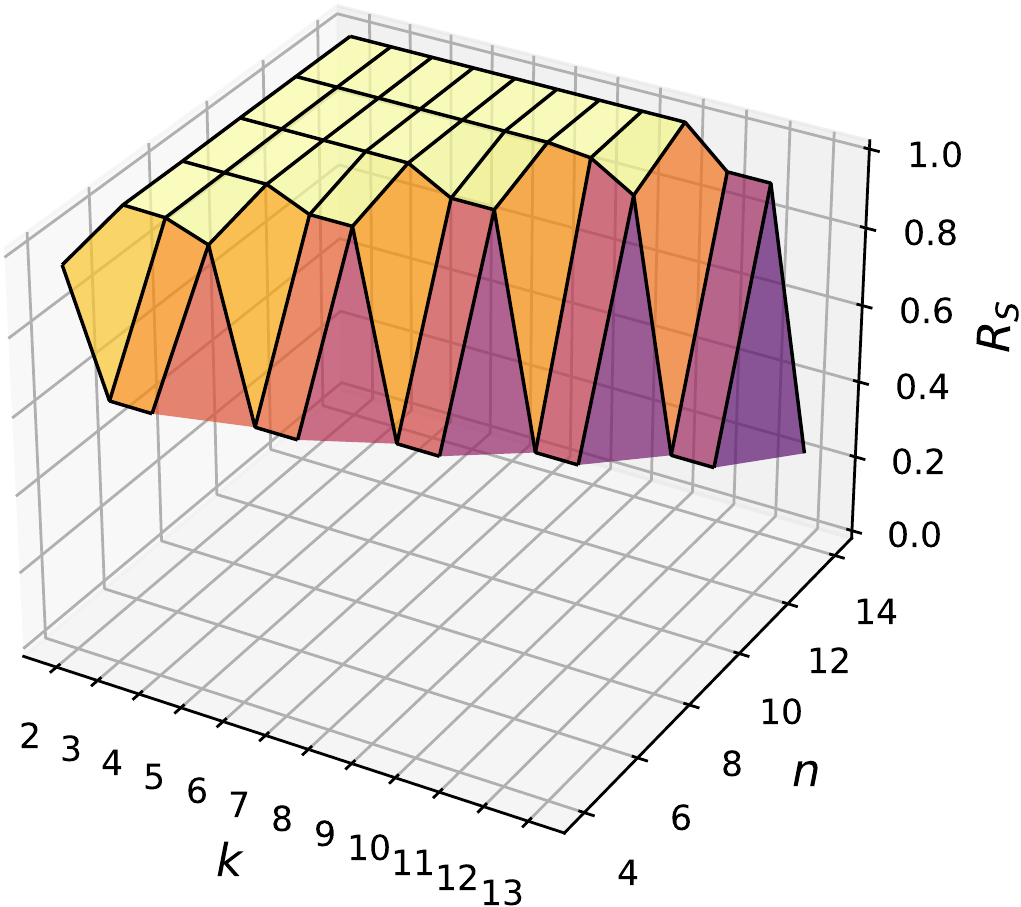}}
		\label{fig:type3_r_0.9_BC2}		
	}
	~
	\subfloat[][$r=0.9$, BC3]
	{
		\centering\resizebox{0.31\textwidth}{!}{\includegraphics{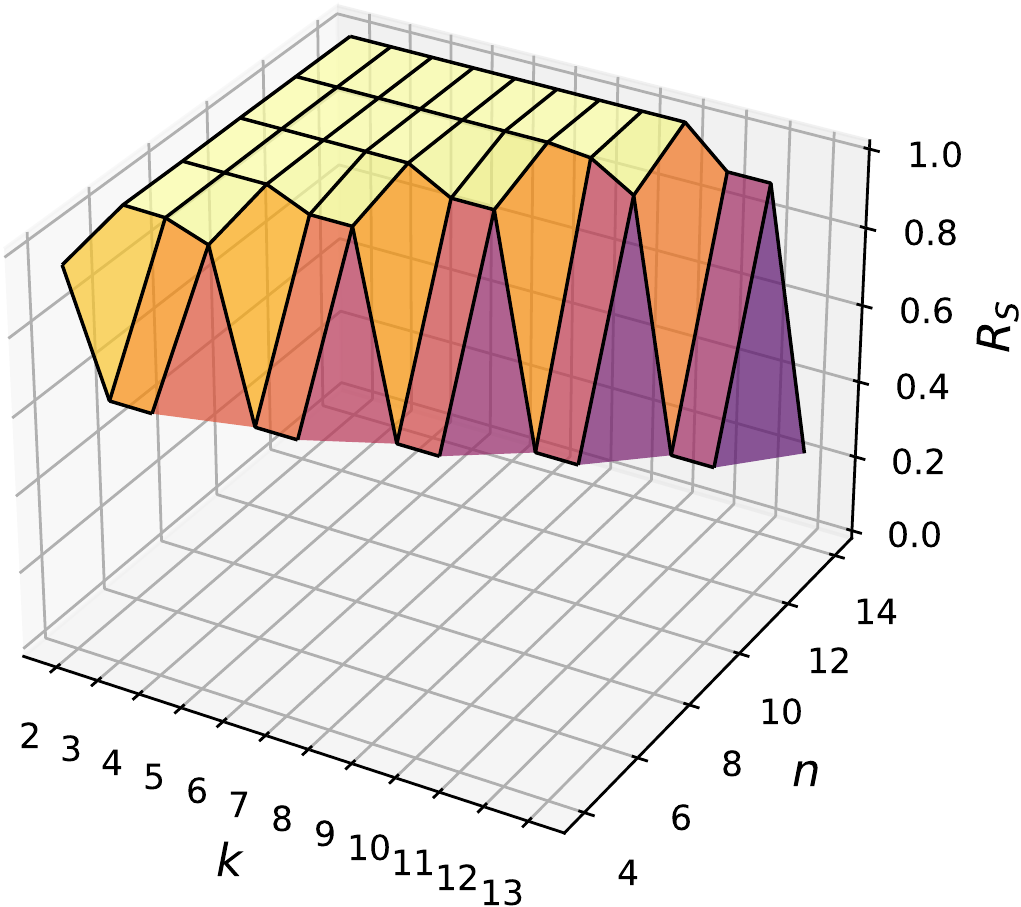}}
		\label{fig:type3_r_0.9_BC3}		
	}
	
	\caption{System reliability plots with varying $k$ and $n$ for the systems with $r=0.5,0.7,0.9$ under each balance condition}
	\label{fig:reliability_comparison_type_3}
    \end{figure}
	
    \subsection{Summary with interpretation}
    In this section, we summarize the experimental results followed by the corresponding interpretations. Firstly, we consistently observed $R_S^{\textrm{BC3}}\ge R_S^{\textrm{BC2}}\ge R_S^{\textrm{BC1}}$ and it demonstrates that the introduction of BC3 can enhance the system reliability of \ckngb systems. The analysis in Section~\ref{subsec:relationship} provides the supporting theory, which shows that BC3 covers BC2 and BC1. Although BC3 is considered as the most generalized balance condition for \ckngb systems to the best of our knowledge, the system reliability can be further improved if one can find another effective balance condition.
	
    Secondly, $R_S$ shows definite ``S''-shaped curve to the unit reliability $r$ when all the other parameters fixed. That is, the unit-wise reliability improvement is more effective in terms of enhancing the system reliability when $r$ is not around the extreme values $0$ or $1$. This tendency can be useful for the impact-effort analysis of unit-wise reliability improvement.
	
    Finally, it is trivial that the increment of $k$ results in the decrement of $R_S$. However, we consistently observed more significant $R_S$'s reduction when $k$ is increased from an even number to an odd number from Figs.~\ref{fig:reliability_comparison_type_2}-\ref{fig:reliability_comparison_type_3}. From this observation, we realize that determining $k$ in a \ckngb system as even as possible must be advantageous for most of the systems.

    \section{Conclusion} \label{sec:conclusion}
    This article examines the reliability of \ckngb systems based on several definitions of the balance condition. Motivated by two previously investigated conditions, BC1 (symmetry) and BC2 (proportionality), we propose a new balance condition BC3 based on the center of gravity concept. Among the three conditions, the newly proposed BC3 is advantageous in terms of its simplicity and generality because both BC1 and BC2 can be regarded as special cases of BC3. We investigate the inclusion relationships among the balance conditions through mathematical proofs and several numerical examples. For system reliability, we apply the minimum tie-set method in which the system is interpreted as a parallel system consists of minimum tie-sets. Since the system reliability is proportional to the number of minimum tie-sets, considering BC3 can discover more minimum tie-sets, thereby enhancing system reliability. Extensive numerical studies demonstrate the effect of considering BC3 and the changing trends of system reliability for varying system parameters.
	
    There are several future research suggestions that extend this paper. First, the reliability evaluation method used in this paper requires complete enumeration of all working states, which can be computationally challenging when $k$ and/or $n$ are large. In this regard, a more efficient method that reasonably approximates system reliability in a shorter time can be explored. Second, assumptions considered throughout this paper can be relaxed. For example, we have assumed that survival and failure probabilities for each unit are constants at any time and that the unit is only of two states (failed and non-failed). The constant probability assumption can be relaxed by introducing a failure time distribution such as exponential and Weibull, etc. Also, considering multi-state units would relax the binary-state assumption, thereby opening up another interesting research direction regarding \ckngb systems.

    \section*{Acknowledgment}
    This work was supported in part by the National Research Foundation of Korea (NRF) grants (No. 2021R1G1A1094924, 2021R1A2C1094699 and 2021R1A4A1031019) funded by the Korea government (Ministry of Science and ICT, MSIT) and in part by Korea Institute for Advancement of Technology (KIAT) grant funded by the Korea Government (MOTIE) (P0008691, HRD Program for Industrial Innovation).
	
    \appendix
    \section*{Appendix}
    \section{Proof of Proposition \ref{prop:b1b3}} \label{app:A}
    Consider a conventional two-dimensional coordinate system with x-axis and y-axis on which the units are circularly and evenly located. Assume that we have a subsystem $U=\{u_1,u_2,...,u_{|U|}\}$ that satisfies BC1. Then, the subsystem $U$ is symmetric with respect to at least a pair of perpendicular axes by the definition~\ref{def:bc1}. That is, for any operating unit $u\in U$, we have another operating unit $u'\in U$ such that the angle between $u$ and $u'$ is $\pi$. Letting $(x_u,y_u)$ be the x-y coordinate of $u$, we have
    \begin{align*}
	\begin{bmatrix}
		x_u \\
		y_u 
	\end{bmatrix}
	= 
	\begin{bmatrix}
		r\cos{\theta} \\
		r\sin{\theta}
	\end{bmatrix}
	,
    \end{align*}
    where $r$ is the Euclidean distance from the origin $(0,0)$ to the unit $u$ and $\theta$ is the angle between x-axis and unit $u$. Since the angle between $u$ and $u'$ is $\pi$, we have
    \begin{align*}
	\begin{bmatrix}
		x_{u'} \\
		y_{u'} 
	\end{bmatrix}
	= 
	\begin{bmatrix}
		r\cos{(\theta+\pi)} \\
		r\sin{(\theta+\pi)}
	\end{bmatrix}
	.
    \end{align*}
    
    Letting $(\bar{x}_{\{u,u'\}},\bar{y}_{\{u,u'\}})$ be the coordinate of a center of gravity formed by the units $u$ and $u'$, we have
    \begin{align*}
	\begin{bmatrix}
		\bar{x}_{\{u,u'\}} \\
		\bar{y}_{\{u,u'\}} 
	\end{bmatrix}
	=
	\begin{bmatrix}
		\frac{1}{2}\{r\cos{(\theta)}+r\cos{(\theta+\pi)}\} \\
		\frac{1}{2}\{r\sin{(\theta)}+r\sin{(\theta+\pi)}\}
	\end{bmatrix}
	=
	\begin{bmatrix}
		-r\left\{\cos{\left(\theta+\frac{\pi}{2}\right)}\cos{\left(\frac{\pi}{2}\right)}\right\} \\
		r\left\{\sin{\left(\theta+\frac{\pi}{2}\right)}\cos{\left(\frac{\pi}{2}\right)}\right\}
	\end{bmatrix}
	=
	\begin{bmatrix}
		0 \\
		0
	\end{bmatrix}
	,
    \end{align*}
    applying the simple trigonometric angle addition formula and because $\cos(\pi/2)=0$.
    
    Since the above result applies for any unit in a subsystem satisfying BC1, the center of gravity of such a subsystem is always $(0,0)$. Therefore, BC1 implies BC3.
	
    \section{Proof of Proposition \ref{prop:b2b3}} \label{app:B}
    \renewcommand{\theequation}{\thechapter.\arabic{equation}}
    \numberwithin{equation}{section}
    Assume that we have a subsystem $U=\{u_1,u_2,...,u_{|U|}\}$ that satisfies BC2. Then, the sector angles are either identical for all $s=1, 2, \dots, N$, as described in remark \ref{remark:BC2}(\ref{remark:BC2a}), or the opposite angles of one another as described in remark \ref{remark:BC2}(\ref{remark:BC2b}). In this proof, we will only show that remark \ref{remark:BC2}(\ref{remark:BC2a}) implies BC3 because remark \ref{remark:BC2}(\ref{remark:BC2b}) is trivially considered as satisfying BC1 (symmetry) hence implies BC3 by proposition \ref{prop:b1b3}.
 
    Given the identical sector angles, we can partition the subsystem $U$ into $N_\textrm{SG}$ mutually exclusive and exhaustive subgroups $U_j$'s each of which have the same number of units such that
    \begin{align*}
        U = \{u_1, u_2, \dots, u_{|U|}\} = \bigcup_{j=1}^{N_\textrm{SG}} U_j\quad\textrm{and}\quad|U_1|=|U_2|=\cdots =|U_{N_\textrm{SG}}|. 
    \end{align*}

    Because all the subgroups are homogeneous and all the sector angles are congruent, we can interpret the subsystem $U$ as an equivalent system $V$ comprised of $N_\textrm{SG}$ \textit{virtually aggregated} units, say units $v_j$'s, such that $V=\{v_1,...,v_{N_\textrm{SG}}\}$ where the angle between any two such units is $\theta=2\pi/N_\textrm{SG}$. Considering the physical properties, we define the x-y coordinate of $v_j$ to be the center of gravity formed by the units in the associated subgroup $U_j$: $(x_{v_j},y_{v_j})\equiv(\bar{x}_{U_j}, \bar{y}_{U_j})$. Then, the center of gravity formed by the subsystem $U$, $(\bar{x}_U,\bar{y}_U)$ can be expressed as follows:
    \begin{align*}
        \begin{bmatrix}
            \bar{x}_U \\
            \bar{y}_U 
	\end{bmatrix}
        =
        \begin{bmatrix}
            \bar{x}_V \\
            \bar{y}_V 
	\end{bmatrix}
        =
        \begin{bmatrix}
		  \frac{1}{N_\textrm{SG}}\sum_{j=1}^{N_\textrm{SG}}x_{v_j}\\
		  \frac{1}{N_\textrm{SG}}\sum_{j=1}^{N_\textrm{SG}}y_{v_j}
	\end{bmatrix}
        =
        \begin{bmatrix}
		  \frac{1}{N_\textrm{SG}}\sum_{j=1}^{N_\textrm{SG}}\bar{x}_{U_j}\\
		  \frac{1}{N_\textrm{SG}}\sum_{j=1}^{N_\textrm{SG}}\bar{y}_{U_j}
        \end{bmatrix}
        .
    \end{align*}
    
    Let $r$ be the Euclidean distance between the origin $(0,0)$ and a virtual unit $v$. Without loss of generality, we adjust the coordinate of unit $v_1$ to $(r,0)$ and arrange all the other units $v_j$'s counterclockwise in ascending order such that the x-y coordinate of unit $v_j$, $(x_{v_j}, y_{v_j})$, becomes $\left(\cos\left(\left(j-1\right)\theta\right),\sin\left(\left(j-1\right)\theta\right) \right)$ for $j=1,..., N_\textrm{SG}$ where $\theta=2\pi/N_{\textrm{SG}}$. 
    Then, $(\bar{x}_U,\bar{y}_U)$ is
     \begin{align}
        \begin{bmatrix}
            \bar{x}_U \\
            \bar{y}_U 
	\end{bmatrix}
        =
        \begin{bmatrix}
		  \frac{1}{N_\textrm{SG}}\sum_{j=1}^{N_\textrm{SG}}r\cos{\left((j-1)\frac{2\pi}{N_\textrm{SG}}\right)}\\
		  \frac{1}{N_\textrm{SG}}\sum_{j=1}^{N_\textrm{SG}}r\sin{\left((j-1)\frac{2\pi}{N_\textrm{SG}}\right)}	
        \end{bmatrix}
        .
        \label{eq:cog1}
    \end{align}   

    To show that the RHS of Eq.~(\ref{eq:cog1}) is indeed the origin $(0,0)$, we will use the following lemma.
    \begin{lemma}[Knapp \cite{K2009}]
        If $a,d\in\mathbb{R}$, $d\neq 0$, and $n$ is a positive integer, then 
        \begin{align}
            \sum_{k=0}^{n-1}\cos{(a+kd)}=\frac{\sin{(nd/2)}}{\sin{(d/2)}} \cos{\left(a+\frac{(n-1)d}{2}\right)} \label{lemma:cos}
        \end{align}
        and
        \begin{align}
            \sum_{k=0}^{n-1}\sin{(a+kd)}=\frac{\sin{(nd/2)}}{\sin{(d/2)}} \sin{\left(a+\frac{(n-1)d}{2}\right)}. \label{lemma:sin}
        \end{align}
    \end{lemma}
    
    First, substituting $k=j-1$, $a=0$, $n=N_\textrm{SG}$, and $d=2\pi/N_\textrm{SG}$ into the both sides of Eq.~(\ref{lemma:cos}), we have the following equation:
    \begin{align}
        \sum_{j=1}^{N_\textrm{SG}}\cos{\left((j-1)\frac{2\pi}{N_\textrm{SG}}\right)}=\frac{\sin{\pi}}{\sin{(\pi/N_\textrm{SG})}} \cos{\left(\frac{(N_\textrm{SG}-1)\pi}{N_\textrm{SG}}\right)}, \label{eq:cog_x}
    \end{align}    
    which reduces to $0$ because we have $\sin{\pi}=0$ in the RHS.

    Similarly, substituting $k=j-1$, $a=0$, $n=N_\textrm{SG}$, and $d=2\pi/N_\textrm{SG}$ into the both sides of Eq.~(\ref{lemma:sin}), we have the following equation:
    \begin{align}
        \sum_{j=1}^{N_\textrm{SG}}\sin{\left((j-1)\frac{2\pi}{N_\textrm{SG}}\right)}=\frac{\sin{\pi}}{\sin{(\pi/N_\textrm{SG})}} \sin{\left(\frac{(N_\textrm{SG}-1)\pi}{N_\textrm{SG}}\right)}, \label{eq:cog_y}
    \end{align}    
    which reduces to $0$ because we have $\sin{\pi}=0$ in the RHS. Combining the results from Eqs.~(\ref{eq:cog_x}) and (\ref{eq:cog_y}), we notice that the RHS of Eq.~(\ref{eq:cog1}) is $[0,0]^\top$. Therefore, BC2 implies BC3.

    \bibliographystyle{plainnat}
    \bibliography{manuscript} 

\begin{thebibliography}{29}
\providecommand{\natexlab}[1]{#1}
\providecommand{\url}[1]{\texttt{#1}}
\expandafter\ifx\csname urlstyle\endcsname\relax
  \providecommand{\doi}[1]{doi: #1}\else
  \providecommand{\doi}{doi: \begingroup \urlstyle{rm}\Url}\fi

\bibitem[Allain(2017)]{A2017}
Rhett Allain.
\newblock {How do drones fly? Physics, of course!}, 5 2017.
\newblock URL \url{https://www.wired.com/2017/05/the-physics-of-drones/}.

\bibitem[Dui et~al.(2021)Dui, Zhang, Bai, and Chen]{DZBC2021}
Hongyan Dui, Chi Zhang, Guanghan Bai, and Liwei Chen.
\newblock Mission reliability modeling of uav swarm and its structure
  optimization based on importance measure.
\newblock \emph{Reliability Engineering \& System Safety}, 215:\penalty0
  107879, 2021.
\newblock ISSN 0951-8320.
\newblock \doi{https://doi.org/10.1016/j.ress.2021.107879}.

\bibitem[Elsayed(2021)]{ELSAYED2021}
Elsayed~A. Elsayed.
\newblock \emph{Reliability Engineering}.
\newblock Wiley Series in Systems Engineering and Management. Wiley, 2021.
\newblock ISBN 9781119665922.
\newblock URL \url{https://books.google.co.kr/books?id=g1BSzQEACAAJ}.

\bibitem[Endharta and Ko(2020)]{EK2020}
Alfonsus~Julanto Endharta and Young~Myoung Ko.
\newblock Economic design and maintenance of a circular $k$-out-of-$n$: G
  balanced system with load-sharing units.
\newblock \emph{IEEE Transactions on Reliability}, 69\penalty0 (4):\penalty0
  1465--1479, 2020.
\newblock \doi{10.1109/TR.2020.2969236}.

\bibitem[Endharta et~al.(2018)Endharta, Yun, and Ko]{EYK2018}
Alfonsus~Julanto Endharta, Won~Young Yun, and Young~Myoung Ko.
\newblock Reliability evaluation of circular $k$-out-of-$n$: G balanced systems
  through minimal path sets.
\newblock \emph{Reliability Engineering \& System Safety}, 180:\penalty0
  226--236, 2018.
\newblock \doi{https://doi.org/10.1016/j.ress.2018.07.023}.

\bibitem[Guo and Elsayed(2019)]{GE2019}
Jingbo Guo and Elsayed~A. Elsayed.
\newblock Reliability of balanced multi-level unmanned aerial vehicles.
\newblock \emph{Computers \& Operations Research}, 106:\penalty0 1--13, 2019.
\newblock ISSN 0305-0548.
\newblock \doi{https://doi.org/10.1016/j.cor.2019.01.013}.

\bibitem[Hao et~al.(2019)Hao, Yeh, Wang, Wang, and Sun]{HYWWS2019}
Zhifeng Hao, Wei-Chang Yeh, Jing Wang, Gai-Ge Wang, and Bin Sun.
\newblock A quick inclusion-exclusion technique.
\newblock \emph{Information Sciences}, 486:\penalty0 20--30, 2019.
\newblock ISSN 0020-0255.
\newblock \doi{https://doi.org/10.1016/j.ins.2019.02.004}.

\bibitem[Heidtmann(1982)]{H1982}
Klaus~D. Heidtmann.
\newblock Improved method of inclusion-exclusion applied to $k$-out-of-$n$
  systems.
\newblock \emph{IEEE Transactions on Reliability}, R-31\penalty0 (1):\penalty0
  36--40, 1982.
\newblock \doi{10.1109/TR.1982.5221218}.

\bibitem[Hirata et~al.(2000)Hirata, Brown, and Shannon]{HBS2000}
Christopher Hirata, Nathan Brown, and Derek Shannon.
\newblock Mars scheme iv: The mars society of caltech human exploration of mars
  endeavor.
\newblock In \emph{Proceedings of the Third International Mars Society
  Convention}, Paper in Track 2 A, pages 1--20, 8 2000.

\bibitem[Hua and Elsayed(2016{\natexlab{a}})]{HE2016a}
Dingguo Hua and Elsayed~A. Elsayed.
\newblock Degradation analysis of $k$-out-of-$n$ pairs: G balanced system with
  spatially distributed units.
\newblock \emph{IEEE Transactions on Reliability}, 65\penalty0 (2):\penalty0
  941--956, 2016{\natexlab{a}}.
\newblock \doi{10.1109/TR.2015.2494683}.

\bibitem[Hua and Elsayed(2016{\natexlab{b}})]{HE2016b}
Dingguo Hua and Elsayed~A. Elsayed.
\newblock Reliability estimation of $k$-out-of-$n$ pairs: G balanced systems
  with spatially distributed units.
\newblock \emph{IEEE Transactions on Reliability}, 65\penalty0 (2):\penalty0
  886--900, 2016{\natexlab{b}}.
\newblock \doi{10.1109/TR.2015.2495153}.

\bibitem[Hua and Elsayed(2018)]{HE2018}
Dingguo Hua and Elsayed~A. Elsayed.
\newblock Reliability approximation of $k$-out-of-$n$ pairs: G balanced systems
  with spatially distributed units.
\newblock \emph{IISE Transactions}, 50\penalty0 (7):\penalty0 616--626, 2018.
\newblock \doi{10.1080/24725854.2018.1431742}.

\bibitem[Huang et~al.(2023)Huang, Xu, Huang, and Fang]{HXHF2023}
Xinqian Huang, Liang Xu, Ying Huang, and Yisong Fang.
\newblock Reliability analysis for k-out-of-n: F load sharing systems operating
  in a shock environment.
\newblock \emph{IEEE Access}, 11:\penalty0 18227--18233, 2023.
\newblock \doi{10.1109/ACCESS.2023.3247449}.

\bibitem[Knapp(2009)]{K2009}
Michael~P. Knapp.
\newblock Sines and cosines of angles in arithmetic progression.
\newblock \emph{Mathematics Magazine}, 82\penalty0 (5):\penalty0 371--372,
  2009.
\newblock \doi{10.4169/002557009X478436}.

\bibitem[McGrady(1985)]{M1985}
Patrick~W. McGrady.
\newblock The availability of a $k$-out-of-$n$: G network.
\newblock \emph{IEEE Transactions on Reliability}, R-34\penalty0 (5):\penalty0
  451--452, 1985.
\newblock \doi{10.1109/TR.1985.5222230}.

\bibitem[Rushdi(1986)]{R1986}
Ali~M. Rushdi.
\newblock Utilization of symmetric switching functions in the computation of
  $k$-out-of-$n$ system reliability.
\newblock \emph{Microelectronics Reliability}, 26\penalty0 (5):\penalty0
  973--987, 1986.
\newblock ISSN 0026-2714.
\newblock \doi{https://doi.org/10.1016/0026-2714(86)90239-8}.

\bibitem[Sarper and Sauer(2002)]{SS2002}
Huseyin Sarper and Wolfang~J. Sauer.
\newblock New reliability configuration for large planetary descent vehicles.
\newblock \emph{Journal of Spacecraft and Rockets}, 39\penalty0 (4):\penalty0
  639--642, 2002.
\newblock \doi{10.2514/2.3856}.

\bibitem[Stamate et~al.(2017)Stamate, Nicolescu, and Pupază]{SNP2017}
Mihai-Alin Stamate, Adrian-Florin Nicolescu, and Cristina Pupază.
\newblock Mathematical model of a multi-rotor drone prototype and calculation
  algorithm for motor selection.
\newblock \emph{Proceedings in Manufacturing Systems}, 12\penalty0
  (3):\penalty0 119--128, 2017.

\bibitem[Tian et~al.(2023)Tian, Yang, Li, and Wang]{TYLW2023}
Tianzi Tian, Jun Yang, Lei Li, and Ning Wang.
\newblock Reliability assessment of performance-based balanced systems with
  rebalancing mechanisms.
\newblock \emph{Reliability Engineering \& System Safety}, 233:\penalty0
  109133, 2023.
\newblock ISSN 0951-8320.
\newblock \doi{https://doi.org/10.1016/j.ress.2023.109133}.

\bibitem[{UAV Systems International}(2023)]{drone_product}
{UAV Systems International}.
\newblock {Heavy lift payload drones}.
\newblock \url{https://uavsystemsinternational.com/pages/heavy-payload-drones},
  2023.
\newblock Accessed: 2023-05-08.

\bibitem[Wang et~al.(2021)Wang, Qiu, Wang, and Lin]{WQWL2021}
Jingjing Wang, Qingan Qiu, Huanhuan Wang, and Cong Lin.
\newblock Optimal condition-based preventive maintenance policy for balanced
  systems.
\newblock \emph{Reliability Engineering \& System Safety}, 211:\penalty0
  107606, 2021.
\newblock ISSN 0951-8320.
\newblock \doi{https://doi.org/10.1016/j.ress.2021.107606}.

\bibitem[Wang et~al.(2022{\natexlab{a}})Wang, Zheng, and Lin]{WZL2022}
Jingjing Wang, Rui Zheng, and Tianran Lin.
\newblock Maintenance modeling for balanced systems subject to two competing
  failure modes.
\newblock \emph{Reliability Engineering \& System Safety}, 225:\penalty0
  108637, 2022{\natexlab{a}}.
\newblock \doi{https://doi.org/10.1016/j.ress.2022.108637}.

\bibitem[Wang et~al.(2022{\natexlab{b}})Wang, Zhao, and Zuo]{WZZ2022}
Siqi Wang, Xian Zhao, and Ming~J. Zuo.
\newblock A multi-state $k$-out-of-$n$: F balanced system with a rebalancing
  mechanism.
\newblock \emph{Quality and Reliability Engineering International}, 38\penalty0
  (6):\penalty0 2908--2920, 2022{\natexlab{b}}.
\newblock \doi{https://doi.org/10.1002/qre.2867}.

\bibitem[Wang et~al.(2020)Wang, Zhao, Wu, and Lin]{WZWL2020}
Xiaoyue Wang, Xian Zhao, Congshan Wu, and Cong Lin.
\newblock Reliability assessment for balanced systems with restricted
  rebalanced mechanisms.
\newblock \emph{Computers \& Industrial Engineering}, 149:\penalty0 106801,
  2020.
\newblock \doi{https://doi.org/10.1016/j.cie.2020.106801}.

\bibitem[Wang et~al.(2023)Wang, Ning, Zhao, and Wu]{WNZW2023}
Xiaoyue Wang, Ru~Ning, Xian Zhao, and Congshan Wu.
\newblock Reliability assessments for two types of balanced systems with
  multi-state protective devices.
\newblock \emph{Reliability Engineering \& System Safety}, 229:\penalty0
  108852, 2023.
\newblock ISSN 0951-8320.
\newblock \doi{https://doi.org/10.1016/j.ress.2022.108852}.

\bibitem[Wu et~al.(2022)Wu, Zhao, Wang, and Song]{WZWS2022}
Congshan Wu, Xian Zhao, Siqi Wang, and Yanbo Song.
\newblock Reliability analysis of consecutive-$k$-out-of-$r$-from-$n$
  subsystems: F balanced systems with load sharing.
\newblock \emph{Reliability Engineering \& System Safety}, 228:\penalty0
  108776, 2022.
\newblock ISSN 0951-8320.
\newblock \doi{https://doi.org/10.1016/j.ress.2022.108776}.

\bibitem[Xing and Johnson(2023)]{XJ2023}
Liudong Xing and Barry~W. Johnson.
\newblock Reliability theory and practice for unmanned aerial vehicles.
\newblock \emph{IEEE Internet of Things Journal}, 10\penalty0 (4):\penalty0
  3548--3566, 2023.
\newblock \doi{10.1109/JIOT.2022.3218491}.

\bibitem[Yangyao et~al.(2023)Yangyao, Xinchen, Tianxiang, and Zijian]{YXTZ2023}
Shi Yangyao, Zhuang Xinchen, Yu~Tianxiang, and Zhang Zijian.
\newblock Multi-state balance system reliability research considering load
  influence.
\newblock \emph{Reliability Engineering \& System Safety}, 233:\penalty0
  109087, 2023.
\newblock ISSN 0951-8320.
\newblock \doi{https://doi.org/10.1016/j.ress.2023.109087}.

\bibitem[Zhao and Wang(2022)]{ZW2022}
Xiujie Zhao and Ziyu Wang.
\newblock Maintenance policies for two-unit balanced systems subject to
  degradation.
\newblock \emph{IEEE Transactions on Reliability}, 71\penalty0 (2):\penalty0
  1116--1126, 2022.
\newblock \doi{10.1109/TR.2022.3167046}.

\end{thebibliography}

\end{document}